\documentclass[conference]{IEEEtran}
\pdfoutput=1
\usepackage[T1]{fontenc}
\usepackage{lmodern}
\usepackage{amsmath} 
\usepackage{amssymb}
\usepackage{amsthm}
\usepackage{bbm}
\usepackage{bm}
\usepackage[normalem]{ulem}
\usepackage{enumerate}
\usepackage{comment}
\usepackage{color}
\usepackage{tikz}
\usepackage{xcolor}
\usepackage[margin=20mm]{geometry}
\usepackage{graphicx}
\usepackage{todonotes}
\usepackage[utf8]{inputenc}
\usepackage[linesnumbered,figure]{algorithm2e}
\usepackage{setspace}
\usepackage{algpseudocode}
\usepackage{paralist}
\usepackage{tikz}
\usepackage{wrapfig}
\usepackage[nolist,printonlyused]{acronym}
\usepackage{etoolbox}
\usepackage{amsmath,amsfonts,amsthm,bm}
\usepackage{multirow}
\usepackage{lipsum}
\usepackage{float}
\usepackage{verbatim}

\usepackage[
expansion=true,
protrusion=true,
kerning=true,
tracking=true,
spacing=true
]{microtype}

\newcommand*\ruleline[1]{\par\noindent\raisebox{.8ex}{\makebox[\linewidth]{\hrulefill\hspace{1ex}\raisebox{-.8ex}{#1}
\hspace{1ex}\hrulefill}}}
\newcommand\rev[1]{\textcolor{black}{#1}}
\newcommand\revo[1]{\textcolor{black}{#1}}
\SetKwIF{If}{ElseIf}{Else}{if}{}{else \hspace{-0.15cm} if}{else}{end if}%
\SetKwFor{While}{while}{}{end while}%
\SetKwRepeat{Do}{do}{while}

\usetikzlibrary{arrows,shapes,decorations.pathmorphing,backgrounds,positioning,fit,matrix}
\usetikzlibrary{calc}
\usetikzlibrary{matrix}
\usetikzlibrary{matrix}

\newtheorem{lemma}{Lemma}
\newtheorem{corollary}{Corollary}
\newtheorem{proposition}{\protect\propositionname}
\newtheorem{theorem}{\protect\theoremname}
\providecommand{\propositionname}{Proposition}
\providecommand{\theoremname}{Theorem}
\theoremstyle{definition}
\newtheorem{definition}{Definition}

\def\varplayersset{\mathcal{N}}
\def\varplayerssetdym{N}
\def\varAPsset{\mathcal{A}}
\def\varAPssetdym{A}
\def\varcloudsset{\mathcal{C}}
\def\varcloudssetdym{C}
\def\varslicesset{\mathcal{S}}
\def\varslicessetdym{S}
\def\varcloud{c}
\def\varplayer{i}
\def\vardatasize{D}
\def\vartaskcomplexity{L}
\def\varAP{a}
\def\varoAP{a^{\prime}}
\def\vardecisionsset{\mathfrak{D}}
\def\varoffloaders{O}

\def\varcomppower{F}
\def\vartimecost{T}
\def\varlocal{l}
\def\varphywirelessrate{R}
\def\varoplayer{j}
\def\varslice{s}
\def\varoslice{s^\prime}
\def\varpolicy{\mbox{\small $\mathcal{P}$}}
\def\varuplinkrate{W}
\def\varuplinkratecoef{w}
\def\varcomppowercoef{w}
\def\varuplinkratecoefvector{\textbf{w}}

\def\vardeviceslicecoef{h}
\def\varcost{C}
\def\varindicatorfunction{I}
\def\varresource{r}

\def\vartimeslot{t}
\def\varbandwidthcoef{b}
\def\varbandwidthcoefvector{\textbf{b}}

\def\varbandwidthslicecoef{b}

\def\varlagrangian{{\mathcal{L}}}
\def\varcongestionweight{q}
\def\varresourcesset{\mathcal{R}}
\def\varpotentialfunction{\Psi}
\def\vartime{T}
\def\varresourcecost{c}
\def\varplayerdistance{d}
\def\varpathlossexponent{\alpha}
\def\varbandwidth{B}
\def\vartransmissionpower{P}
\def\varnoisepower{N_0} 
\def\varcomplexityperbit{X}
\def\varnumber{n}

\def\varedgeresourcesset{\mathcal{E}}
\def\varedgeconstant{E}
\def\varedgeresource{e}
\def\varoedgeresource{e^{\prime}}

\newbool{NoLearning}
\booltrue{NoLearning}
\newbool{NLBAGFirst}
\booltrue{NLBAGFirst}
\ifbool{NoLearning}{
\def\varactivedecisionsset{\mathfrak{D}}
\def\vardecisionsset{\mathfrak{D}}
\def\vardecisionsvector{\textbf{d}}
\def\vardecision{d}
\def\varactiveplayersset{\mathcal{N}}

\def\varactiveresourcesset{\mathcal{R}}
\def\varthegame{\Gamma}
}{
\def\varactivedecisionsset{\mathfrak{D}^{\vartimeslot}}
\def\vardecisionsvector{\textbf{d}^{\vartimeslot}}
\def\vardecision{d^{\vartimeslot}}
\def\varactiveplayersset{\mathcal{N}_{\vartimeslot}}

\def\varactiveresourcesset{\mathcal{R}_{\vartimeslot}}
\def\varthegame{\Gamma^\vartimeslot}
}

\def\argmax{\mathop{\rm arg\,max}}
\def\argmin{\mathop{\rm arg\,min}}
\begin{acronym}[MEC]%\addtolength{\itemsep}{-0.5\baselineskip}
\acro{ES-COG}{\emph{edge slicing computation offloading game}}
\acro{SL-COG}{\emph{slice level computation offloading game}}
\acro{NL-BAG}{\emph{network level bandwidth allocation game}}
\acro{OSA-COG}{\emph{optimal slices computation offloading game}}
\acro{COS}{\emph{choose offloading slice}}
\acro{JSS-ERM}{\emph{Joint Slice Selection and Edge Resource Management}}
\acro{SRO}{\emph{slice resource orchestrator}} 
\acro{NSO-ERA}{\emph{network slice orchestration and edge resources allocation}}
\end{acronym}

\begin{document}

%\title{Edge Cloud Resource Management for Selfish Computation Offloading under Network Slicing}
\title{Joint Wireless and Edge Computing Resource Management with Dynamic Network Slice Selection}
%\title{Joint Management of Edge Computing Resources and Dynamic Network Slice Selection}
\author{Sla\dj ana Jo\v{s}ilo and Gy\"orgy D\'an\\
  Division of Network and Systems Engineering,\\
 School of Electrical Engineering and Computer Science\\
  KTH, Royal Institute of Technology, Stockholm, Sweden
  E-mail: \{josilo, gyuri\}@kth.se 
 % \thanks{
  %The work was partly funded by the Swedish Research Council through project 621-2014-6.
  %}
  }

\maketitle

\begin{abstract}
Network slicing is a promising approach for enabling low latency computation offloading in edge computing systems.
In this paper, we consider an edge computing system under network slicing in which the wireless devices generate latency sensitive computational tasks. 
We address the problem of joint dynamic assignment of computational tasks to slices, management of radio resources across slices and management of radio and computing resources within slices. 
We formulate the \ac{JSS-ERM} problem as a mixed-integer problem with the objective to minimize the completion time of computational tasks. 
We show that the \ac{JSS-ERM} problem is NP-hard and develop an approximation algorithm with bounded approximation ratio based on a game theoretic treatment of the problem. 
We provide extensive simulation results to show that network slicing can improve the system performance compared to no slicing and that the proposed solution can achieve significant gains compared to the equal slicing policy. Our results also show that the computational complexity of the proposed algorithm is approximately linear in the number of devices.
%We consider the problem of computation offloading under network slicing in an edge computing system, in which a network operator allocates resource among slices, while slices allocate bandwidth and computing resources to autonomous users who aim at minimizing their task completion times. We provide a game theoretical model of the interaction between the operator, the slices and the users. 
%We derive closed form expressions for the equilibrium bandwidth and computing resource allocation policies of the operator and of the slices. 
%We propose an efficient decentralized algorithm for users to compute an equilibrium offloading strategy under the equilibrium policies, and provide a bound on the price of anarchy. Extensive simulation results show that network slicing can improve the system performance compared to no slicing and that the proposed slicing policy for the network operator can achieve significant gains compared to the equal slicing policy. Our results also show that the time needed to compute an equilibrium offloading strategy is approximately linear in the number of users.
\end{abstract}


\IEEEpeerreviewmaketitle

\section{Introduction}
\label{sec::intro}
Network slicing is emerging as an enabler for providing logical networks that are customized to meet the needs of different kinds of applications, mostly in 5G mobile networks. Horizontal network slices are designed for specific classes of applications, e.g., streaming visual analytics, real-time control, or media delivery,  while vertical network slices are designed for specific industries. Slicing is expected to allow flexible and efficient end-to-end provisioning of bandwidth, composition of in-network processing, e.g., in the form of service chains composed of virtual network functions (VNF), and the allocation of dedicated computing resources. At the same time it provides performance isolation.  Slicing is particularly appealing in combination with edge computing, as network slicing could allow low latency access to customized computing services located in edge clouds~\cite{ordonez2017network,kekki2018mec}. 

Flexibility in network slicing is achieved through service orchestration. Orchestration focuses on the deployment and service-aware adaptation of VNFs and edge cloud services based on predicted workloads. Recent works in the area addressed the joint placement and routing of service function chains, formulated as a virtual network embedding problem~\cite{Rost:2019:VNE:3370562.3370587}, and the problem of joint resource dimensioning and routing~\cite{8850061,8082507}. Typical objectives are maximization of the service capacity or profit under physical (bandwidth and computational power) resource constraints, or the minimization of the energy consumption subject to satisfying service demand. 

Common to the works on service orchestration is that they assume that each application is mapped to a specific slice deterministically, and assume a static resource pool per slice so as to ensure performance isolation~\cite{Rost:2019:VNE:3370562.3370587,8850061,8082507}. A deterministic mapping is, however, not mandatory in practice. While there may be a designated (default) slice for every application, most proposed architectures for network slicing define a set of allowed slices, and the assignment of an application to a slice can be decided dynamically based on the current workload and SLA requirements~\cite{461563398f1a4f8c876135ad7e7d5bed}. The dynamic assignment of applications to slices thus results in a mixture of workloads in the slices,
and consequently calls for flexibility in allocating resources to slices.

The importance of resource management across slices has been widely accepted in the case of the radio access network (RAN)~\cite{461563398f1a4f8c876135ad7e7d5bed}. Such inter-slice resource allocation should happen at short time scales, taking into account slice-level service level agreements (SLAs) and technological constraints (e.g., available RAN technology, such as 5GNR or WiFi-Lic). Recent work in the area has focused on system aspects of virtualizing RANs~\cite{Foukas:2017:ORS:3117811.3117831}, and on the allocation of virtual resource block groups to slices so as to maximize efficiency~\cite{8407021}, but has not considered of the potential impact of inter-slice resource management on service orchestration and on the dynamic assignment of applications to slices. It is thus so far unclear how to perform joint resource management within and across slices, considering the orchestration of communication and computing resources simultaneously.

%In this paper we address the problem of joint dynamic slice selection and inter-slice radio resource management for latency sensitive workloads, and make three important contributions. First, we develop a system model and formulate the joint slice selection and radio resource management (SSRRM) problem, and show it is NP-hard. 
\rev{In this paper we address the problem of joint dynamic slice selection, inter-slice radio resource management and intra-slice radio and computing resource management for latency sensitive workloads, and make three important contributions.}
First, we formulate the \rev{joint slice selection and edge resource management (JSS-ERM) problem,} and show that it is NP-hard.
Second, we analyze the optimal solution structure, and we develop an efficient approximation algorithm with bounded approximation ratio inspired by a game theoretic treatment of the problem.
%, proving equilibrium existence in the resulting game and providing a bound on the price of anarchy. 
Third, we provide extensive numerical results to show that the resulting system performance significantly outperforms baseline resource allocation policies. 

The rest of the paper is organized as follows. Section \ref{sec::model} introduces the system model and Section~\ref{sec::_optimization_problem} the problem formulation. Sections~\ref{sec::game} and ~\ref{sec::slice_level} provide the analytical results, and Section~\ref{sec::numerical} shows numerical results. Section~\ref{sec::related} discusses related work and Section~\ref{sec::conclusion} concludes the paper. 
 
\section{System Model}
\label{sec::model}
We consider a slicing enabled mobile backhaul including mobile edge computing (MEC) resources that serves a set $\varplayersset \!\!=\!\! \{1,2,\hdots,\varplayerssetdym\}$ of wireless devices (WDs) that generate computationally intensive tasks. WDs can offload their tasks through a set $\varAPsset \!\!=\!\! \{1,2,\hdots,\varAPssetdym\}$ of access points (APs) to a set $\varcloudsset \!\!=\! \{1,2,\hdots,\varcloudssetdym\}$ of edge clouds (ECs). \rev{APs and ECs form the set $\varedgeresourcesset \triangleq \varAPsset \cup \varcloudsset$ of edge resources}. We denote by $\varslicesset \!=\!\! \{1,2,\hdots,\varslicessetdym\}$ the set of slices in the network, which include certain combinations of computing resources (e.g., CPUs, GPUs, NPUs and/or FPGAs), optimized for executing some types of tasks.
%Each slice is managed by a provider and it offers certain types of computing resources (e.g. CPUs, GPUs, NPUs and/or FPGAs), optimized for executing different types of tasks. 
  
  We characterize a task generated by WD $\varplayer$ by the size $\vardatasize_\varplayer$ of the input data and by its complexity, which we define as the \emph{expected} number of instructions required to perform the computation. 
 %Similar to other works~\cite{zheng2018dynamic,chen2014decentralized}, we consider that the parameters $\vardatasize_\varplayer$ and $\vartaskcomplexity_{\varplayer}$ can be estimated from the measurements by applying the methods described in~\cite{neto2018uloof,chun2011clonecloud,cuervo2010maui}.
 Since the WDs and the slices may have different instruction set architectures, the number of instructions required to execute the same task may also differ. Hence, for a task generated by WD $\varplayer$ we denote by $\vartaskcomplexity_{\varplayer}$ and $\vartaskcomplexity_{\varplayer,\varslice}$ the expected number of instructions required to perform the computation locally and in slice $\varslice$, respectively. 
 Similar to other works~\cite{zheng2018dynamic,chen2014decentralized,jovsilo2018decentralized}, we consider that $\vardatasize_\varplayer$, $\vartaskcomplexity_{\varplayer}$ and  $\vartaskcomplexity_{\varplayer,\varslice}$ can be estimated from measurements by applying the methods described in~\cite{neto2018uloof,chun2011clonecloud,cuervo2010maui}. 

\rev{We consider that each WD $\varplayer$ generates a computational task at a time; each task is atomic and can be either offloaded for computation or performed locally on the WD it was generated at. In the case of offloading, the WD will be assigned to exactly one slice $\varslice \in \varslicesset$ and within the slice to exactly one AP $\varAP \in \varAPsset$ and to exactly one EC $\varcloud \in \varcloudsset$. Therefore, we define the set of feasible decisions for WD $\varplayer$ as  $\varactivedecisionsset_\varplayer \triangleq  \{\varplayer\} \cup \{(\varAP,\varcloud,\varslice)| \varAP \in \varAPsset, \varcloud \in \varcloudsset, \varslice \in \varslicesset\}$ and we use variable $\vardecision_\varplayer \in \varactivedecisionsset_\varplayer$ to indicate the decision for WD $\varplayer$'s task (i.e., $\vardecision_\varplayer =\varplayer$ indicates that WD $\varplayer$ performs the task locally and $ \vardecision_\varplayer = (\varAP,\varcloud,\varslice)$ indicates that WD $\varplayer$ should offload its task through AP $\varAP$ to  EC $\varcloud$ in slice $\varslice$). Furthermore, we define a decision vector $\vardecisionsvector \triangleq (\vardecision_\varplayer)_{\varplayer \in \varactiveplayersset}$ as the collection of the decisions of all WDs and we define the set $\varactivedecisionsset \triangleq \times_{\varplayer \in \varplayersset} \varactivedecisionsset_\varplayer$, i.e., the set of all possible decision vectors.}

\rev{For a decision vector $\vardecisionsvector \in \vardecisionsset$ we define the set $\varoffloaders_{(\varAP,\varslice)}(\vardecisionsvector) \triangleq \{\varplayer \in \varplayersset| \vardecision_\varplayer = (\varAP,\cdot,\varslice)\}$ of all WDs that use AP $\varAP$ in slice $\varslice$ and the set $\varoffloaders_{\varAP}(\vardecisionsvector) = \cup_{\varslice \in \varslicesset}\varoffloaders_{(\varAP,\varslice)}(\vardecisionsvector)$ of all WDs that use AP $\varAP$. Similarly, we define the set $\varoffloaders_{(\varcloud,\varslice)}(\vardecisionsvector) \triangleq \{\varplayer \in \varplayersset| \vardecision_\varplayer = (\cdot,\varcloud,\varslice)\}$ of all WDs that use EC $\varcloud$ in slice $\varslice$ and the set $\varoffloaders_{\varcloud}(\vardecisionsvector) = \cup_{\varslice \in \varslicesset}\varoffloaders_{(\varcloud,\varslice)}(\vardecisionsvector)$ of all WDs that use EC $\varcloud$. Finally, we define the local computing singleton set $\varoffloaders_{\varplayer}(\vardecisionsvector) \subset \{\varplayer, \emptyset\}$ for WD $\varplayer$ (i.e., $\varoffloaders_\varplayer(\vardecisionsvector) = \{\varplayer\}$ when WD $\varplayer$ performs the computation locally and $\varoffloaders_\varplayer(\vardecisionsvector) = \emptyset$ otherwise) and the set $\varoffloaders_\varlocal(\vardecisionsvector) = \cup_{\varplayer \in \varplayersset}\varoffloaders_\varplayer(\vardecisionsvector)$ of all WDs that perform the computation locally.}

\begin{figure}[t]
	%\begin{minipage}{0.5\textwidth}		
	\begin{center}
		\includegraphics[width=0.8\columnwidth]{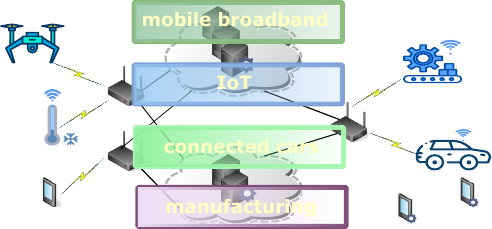}
		\caption{An example of a slicing enabled MEC system that consists of $\varplayerssetdym = 7$ WDs, $ \varcloudssetdym= 2$ ECs and $\varAPssetdym = 3$ APs and $\varslicessetdym = 4$ slices.}
		\label{fig::model}
	\end{center}
	%\end{minipage}
	\vspace{-0.2cm}
\end{figure}
Figure~\ref{fig::model} shows an example of a slicing enabled MEC system that consists of $\varplayerssetdym = 7$ WDs, $ \varcloudssetdym= 2$ ECs and $\varAPssetdym = 3$ APs and $\varslicessetdym = 4$ slices. In this example we have that $2$ out of $7$ WDs perform the computation locally and $5$ out of $7$ WDs offload their tasks. In what follows we discuss our models of communication and computing resources.

\subsection{Communication Resources} 
Communication resources in the system are managed at two levels: at the network level and at the slice level. 

\rev{At the network level, the radio resources of each AP $\varAP \in \varAPsset$ are shared across the slices according to the \emph{inter-slice radio resource allocation policy} $\varpolicy_{\varbandwidthcoef}\!:\!\varactivedecisionsset \!\rightarrow \mathbb{R}_{[0,1]}^{|\varAPsset|\times|\varslicesset|}$, which determines the inter-slice radio resource provisioning coefficients $\varbandwidthslicecoef_{\varAP}^{\varslice} \!\in\! [0,1]$, $\forall (\varAP,\varslice) \!\in\! \varAPsset \times \varslicesset$ such that $\sum_{\varslice \in \varslicesset}\varbandwidthslicecoef_{\varAP}^{\varslice} \leq 1$, $\forall \varAP \in \varAPsset$. %We denote by $\varbandwidthcoefvector \triangleq (\varbandwidthcoef_\varAP^\varslice)_{\varslice \in \varslicesset,\varAP \in \varAPsset}$ the collection of inter-slice radio provisioning coefficients.
}

\rev{At the slice level, the radio resources assigned to each slice $\varslice \in \varslicesset$ are shared among the WDs according to an \emph{intra-slice radio resource allocation policy} $\varpolicy_{\varuplinkratecoef_{\varAP}}^{\varslice}:\varactivedecisionsset \rightarrow \mathbb{R}_{[0,1]}^{|\varAPsset|\times|\varactiveplayersset|}$, which determines the intra-slice radio resource provisioning coefficients $\varuplinkratecoef_{\varplayer,\varAP}^{\varslice} \in [0,1]$, $\forall \varAP \in \varAPsset$ and $\forall \varplayer \in \varoffloaders_{(\varAP,\varslice)}(\vardecisionsvector)$ such that $\sum_{\varplayer \in \varoffloaders_{(\varAP,\varslice)}(\vardecisionsvector)}\varuplinkratecoef_{\varplayer,\varAP}^{\varslice} \leq 1$, $\forall (\varAP,\varslice) \in \varAPsset \times \varslicesset$. %We denote by $\varuplinkratecoefvector^{\varslice} \triangleq (\varuplinkratecoef_{\varplayer,\varAP}^{\varslice})_{\varplayer \in \varactiveplayersset, \varAP \in \varAPsset}$ the collection of intra-slice radio provisioning coefficients of slice $\varslice$.
}

We denote by $\varphywirelessrate_{\varplayer,\varAP}$ the achievable PHY rate of WD $\varplayer$ at AP $\varAP$. $\varphywirelessrate_{\varplayer,\varAP}$  depends on physical signal characteristics, such as path loss and fading, and on the modulation-coding scheme. Given $\varphywirelessrate_{\varplayer,\varAP}$ we can express the actual uplink rate of WD $\varplayer$ at AP $\varAP$ in slice $\varslice$ as 
\begin{equation}\label{eq::uplink_rate}
\varuplinkrate_{\varplayer,\varAP}^{\varslice}(\vardecisionsvector,
\varpolicy_\varbandwidthcoef,\varpolicy_{\varuplinkratecoef_{\varAP}}^{\varslice}) = \varbandwidthcoef_{\varAP}^{\varslice} \varuplinkratecoef_{\varplayer,\varAP}^{\varslice} \varphywirelessrate_{\varplayer,\varAP}.
\end{equation}
%The uplink rate~(\ref{eq::uplink_rate}) together with the input data size $\vardatasize_\varplayer$ determines the transmission time $\vartime_{\varplayer,(\varAP,\varslice)}^{tx}(\vardecisionsvector,
%\varpolicy_\varbandwidthcoef,\varpolicy_{\varuplinkratecoef_{\varAP}}^{\varslice})$ of WD $\varplayer$ that chooses $\vardecision_\varplayer = (\varAP,\cdot,\varslice)$ in strategy profile $\vardecisionsvector$ 
\rev{The uplink rate~(\ref{eq::uplink_rate}) together with the input data size $\vardatasize_\varplayer$ determines the transmission time of WD $\varplayer \in \varoffloaders_{(\varAP,\varslice)}(\vardecisionsvector)$},
\begin{equation}\label{eq::transmission_time}
    \vartime_{\varplayer,\varAP}^{tx,\varslice}(\vardecisionsvector,
\varpolicy_\varbandwidthcoef,\varpolicy_{\varuplinkratecoef_{\varAP}}^{\varslice}) = \frac{\vardatasize_\varplayer}{\varuplinkrate_{\varplayer,\varAP}^{\varslice}(\vardecisionsvector,
\varpolicy_\varbandwidthcoef,\varpolicy_{\varuplinkratecoef_{\varAP}}^{\varslice})}.
\end{equation}
Similar to previous works~\cite{chen2014decentralized,josilo2018selfish,huang2012dynamic,josilo2019tcc} we make the assumption that the time needed to transmit the results of the computation from the EC to the WD can be neglected \rev{because for many applications (e.g., face recognition and tracking) the size of the output data is significantly
smaller than the size $\vardatasize_\varplayer$ of the input data.}

\subsection{Computing Resources} 
Our system model distinguishes between edge cloud resources and local computing resources.

\subsubsection{Edge Cloud Resources}
\rev{We consider that each slice $\varslice \in \varslicesset$ is equipped with a certain combination of computing resources optimized for executing specific types of tasks (e.g, CPUs, GPUs, NPUs, FPGAs), and we denote by $\varcomppower_{\varcloud}^{\varslice}$ the computing capability of EC $\varcloud$ in slice $\varslice$. The computing resources within a slice are shared among the WDs according to the \emph{intra-slice computing power allocation policy} $\varpolicy_{\varcomppowercoef_{\varcloud}}^{\varslice}:\varactivedecisionsset \rightarrow \mathbb{R}_{[0,1]}^{|\varcloudsset|\times|\varactiveplayersset|}$, which determines the intra-slice computing power provisioning coefficients $\varcomppowercoef_{\varplayer,\varcloud}^{\varslice} \in [0,1]$, $\forall \varcloud \in \varcloudsset$ and $\forall \varplayer \in \varoffloaders_{(\varcloud,\varslice)}(\vardecisionsvector)$ such that $\sum_{\varplayer \in \varoffloaders_{(\varcloud,\varslice)}(\vardecisionsvector)}\varcomppowercoef_{\varplayer,\varcloud}^{\varslice} = 1$, $\forall (\varcloud,\varslice) \in \varcloudsset \times \varslicesset$. %We denote by $\varcomppowercoefvector^{\varslice} \triangleq (\varcomppowercoef_{\varplayer,\varcloud}^{\varslice})_{\varplayer \in \varactiveplayersset, \varcloud \in \varcloudsset}$ the collection of computing power provisioning coefficients of slice $\varslice$.
}

Given the computing capability $\varcomppower_{\varcloud}^{\varslice}$ we can express the computing capability allocated to WD $\varplayer$ in EC $\varcloud$ in slice $\varslice$ as 
%We consider that each slice $\varslice$ manages its computing resources according to a \emph{computing resource allocation policy} $\varpolicy_{\varcomppowercoef}^{\varslice}:\varactivedecisionsset \rightarrow \mathbb{R}_{[0,1]}^{|\varcloudsset|\times|\varactiveplayersset|}$. Policy $\varpolicy_{\varcomppowercoef}^{\varslice}$ sets computing power provisioning coefficients $\varcomppowercoef_{\varplayer,\varcloud}^{\varslice} \in [0,1]$, $\forall \varcloud \in \varcloudsset$ and $\forall \varplayer \in \varoffloaders_{(\varcloud,\varslice)}(\vardecisionsvector)$ such that $\sum_{\varplayer \in \varoffloaders_{(\varcloud,\varslice)}(\vardecisionsvector)}\varcomppowercoef_{\varplayer,\varcloud}^{\varslice} = 1$, $\forall (\varcloud,\varslice) \in \varcloudsset \times \varslicesset$. We denote by $\varcomppowercoefvector^{\varslice} \triangleq (\varcomppowercoef_{\varplayer,\varcloud}^{\varslice})_{\varplayer \in \varactiveplayersset, \varcloud \in \varcloudsset}$ the collection of coefficients of slice $\varslice$.
%Given the total computing capability $\varcomppower_{\varcloud}^{\varslice}$ of EC $\varcloud$ in slice $\varslice$, we express the computing capability allocated to WD $\varplayer$ in EC $\varcloud$ in slice $\varslice$ as 
\begin{equation}\label{eq::cloud_to_user_power}
\varcomppower_{\varplayer,\varcloud}^{\varslice}(\vardecisionsvector,\varpolicy_{\varcomppowercoef_{\varcloud}}^{\varslice}) = \varcomppowercoef_{\varplayer,\varcloud}^{\varslice}\varcomppower_{\varcloud}^{\varslice}.
\end{equation}
In order to account for the diversity of computing resources provided by different slices we use the coefficient $\vardeviceslicecoef_{\varplayer,\varslice} \in \mathbb{R}_{\geq 0}$ to capture how well a slice $\varslice$ is tailored for executing a task generated by WD $\varplayer$ and we express the expected number of instructions $\vartaskcomplexity_{\varplayer,\varslice}$ required to execute a task generated by WD $\varplayer$ in slice $\varslice$ as \rev{$\vartaskcomplexity_{\varplayer,\varslice} = \vartaskcomplexity_\varplayer/\vardeviceslicecoef_{\varplayer,\varslice}$ (i.e., a high $\vardeviceslicecoef_{\varplayer,\varslice}$ indicates that a task generated by WD $\varplayer$ is a good match for the computing resources in slice $\varslice$).
%$\vartaskcomplexity_{\varplayer,\varslice} = \vardeviceslicecoef_{\varplayer,\varslice} \vartaskcomplexity_\varplayer$. 
Thus, in our model the computing capability (\ref{eq::cloud_to_user_power}) together with the expected task complexity $\vartaskcomplexity_{\varplayer,\varslice}$  determines the task execution time of WD $\varplayer \in \varoffloaders_{(\varcloud,\varslice)}(\vardecisionsvector)$ as}
%If WD $\varplayer$  chooses $\vardecision_\varplayer = (\cdot,\varcloud,\varslice)$ in strategy profile $\vardecisionsvector$, its expected task complexity $\vartaskcomplexity_{\varplayer,\varslice}$ together with the computing capability (\ref{eq::cloud_to_user_power}) determines the task execution time   
\begin{equation}\label{eq::cloud_execution_time}
    \vartime_{\varplayer,\varcloud}^{ex,\varslice}(\vardecisionsvector,
\varpolicy_\varbandwidthcoef,\varpolicy_{\varuplinkratecoef_{\varAP}}^{\varslice}) = \frac{\vartaskcomplexity_{\varplayer,\varslice}}{\varcomppower_{\varplayer,\varcloud}^{\varslice}(\vardecisionsvector,\varpolicy_{\varcomppowercoef_{\varcloud}}^{\varslice})}.
\end{equation}
%\begin{equation}\label{eq::cloud_execution_time}
%    \vartime_{\varplayer,\varcloud}^{ex,\varslice}(\vardecisionsvector,
%\varpolicy_\varbandwidthcoef,\varpolicy_{\varuplinkratecoef_{\varAP}}^{\varslice}) = \vardeviceslicecoef_{\varplayer,\varslice}\vartaskcomplexity_\varplayer/\varcomppower_{\varplayer,\varcloud}^{\varslice}(\vardecisionsvector,\varpolicy_\varcomppowercoef^\varslice).
%\end{equation}

\subsubsection{Local Computing Resources}
We denote by $\varcomppower_{\varplayer}^{\varlocal}$ the computing capability of WD $\varplayer$ and we express the local execution time $\vartime_{\varplayer}^{ex}$ of WD $\varplayer$ as 
\begin{equation}\label{eq::local_execution_time}
\vartime_{\varplayer}^{ex} = \frac{\vartaskcomplexity_\varplayer}{\varcomppower_{\varplayer}^{\varlocal}}.
\end{equation}
%\begin{equation}\label{eq::cloud_to_user_power}
%\varcomppower_{\varplayer,\varcloud}^{\varslice}(\vardecisionsvector,\varcomppowercoefvector_{\varcloud}^{\varslice}) = \varcomppower_{\varcloud}^{\varslice}\frac{\varcomppowercoef_{\varplayer,\varcloud}^{\varslice}}{\sum_{\varoplayer\in\varoffloaders_{(\varcloud,\varslice)}(\vardecisionsvector)}{\varcomppowercoef_{\varoplayer,\varcloud}^{\varslice}}}.
%\end{equation}

\subsection{Cost Model}
\rev{We define the system cost as the aggregate completion time of all WDs. Before providing a formal definition, we introduce the shorthand notation
\begin{eqnarray}\label{eq::edge_resource_constant}
  &\hspace{-0.2cm}\varedgeconstant_{\varplayer,\varedgeresource}^\varslice =
	\left\{\!\!\! \begin{array}{ll}
	\frac{\vardatasize_\varplayer}{\varphywirelessrate_{\varplayer,\varedgeresource}}& \text{ if }
\varplayer \in \varplayersset, \varedgeresource \in \varedgeresourcesset \cap \varAPsset, \varslice \in \varslicesset\\
	\frac{\vartaskcomplexity_{\varplayer,\varslice}}{\varcomppower_{\varedgeresource}^\varslice}& \text{ if }
\varplayer \in \varplayersset, \varedgeresource \in \varedgeresourcesset \cap \varcloudsset, \varslice \in \varslicesset,\\	
	\end{array} \right.
\end{eqnarray}
\begin{eqnarray}\label{eq::edge_resource_variable}
  &\hspace{-0.2cm}
	\varbandwidthcoef_{\varedgeresource}^\varslice =
	\left\{\!\!\! \begin{array}{ll}
	\varbandwidthcoef_{\varedgeresource}^\varslice& \text{ if }
\varedgeresource \in \varedgeresourcesset \cap \varAPsset, \varslice \in \varslicesset\\
	1& \text{ if }
\varedgeresource \in \varedgeresourcesset \cap \varcloudsset, \varslice \in \varslicesset.\\	
	\end{array} \right.
\end{eqnarray}}

{\bf Cost of WD $\varplayer$:} When offloading, the task completion time consists of two parts: the time needed to transmit the data pertaining to a task through an AP and the time needed to execute a task in an EC. In the case of local computing, the task completion time depends only on the local execution time. Therefore, the cost of WD $\varplayer$ can be expressed as
\begin{equation}\label{eq::user_i_cost}
	\varcost_\varplayer(\vardecisionsvector,\varpolicy_\varbandwidthcoef,\varpolicy_{\varuplinkratecoef_{\varAP}}^{\varslice},\varpolicy_{\varcomppowercoef_{\varcloud}}^{\varslice}) \!=\!
	\left\{\!\!\! \begin{array}{ll}
	\frac{\varedgeconstant_{\varplayer,\varAP}^\varslice}{\varbandwidthcoef_\varAP^\varslice\varuplinkratecoef_{\varplayer,\varAP}^\varslice} \!+\! \frac{\varedgeconstant_{\varplayer,\varcloud}^\varslice}{\varcomppowercoef_{\varplayer,\varcloud}^\varslice},\!&
\!\!\!\!\varindicatorfunction_{\{\vardecision_\varplayer \!=\! (\varAP,\varcloud,\varslice) \}} \!=\! 1,\\
	\vartime_{\varplayer}^{ex},&\!\!\!\!\varindicatorfunction_{\{\vardecision_\varplayer = \varplayer \}} = 1.\\	
	\end{array} \right.\hspace{-0.4cm}	
\end{equation}
where $\varindicatorfunction_{\{\vardecision_\varplayer = \vardecision \}} = 1$ if $\vardecision_\varplayer = \vardecision$ and $\varindicatorfunction_{\{\vardecision_\varplayer=\vardecision\}} = 0$ otherwise.

 {\bf Cost per slice:} We express the cost in slice $\varslice$ as
\begin{equation}\label{eq::cost_per_slice}
  \varcost^{\varslice}(\vardecisionsvector,\varpolicy_\varbandwidthcoef,\varpolicy_{\varuplinkratecoef_{\varAP}}^{\varslice},\varpolicy_{\varcomppowercoef_{\varcloud}}^{\varslice}) \!=\!
  \sum \limits_{\varedgeresource \in \varedgeresourcesset}\sum \limits_{\varplayer \in \varoffloaders_{(\varedgeresource,\varslice)}(\vardecisionsvector)}
   \frac{\varedgeconstant_{\varplayer,\varedgeresource}^\varslice}{\varbandwidthcoef_\varedgeresource^\varslice \varuplinkratecoef_{\varplayer,\varedgeresource}^\varslice}.
\end{equation}

{\bf System cost:} Finally, we express the system cost as
\begin{equation}\label{eq::system_cost}
  \varcost(\vardecisionsvector,\varpolicy_{\varbandwidthcoef},\!\varpolicy_{\!\varuplinkratecoef_\varAP}\!,\!\varpolicy_{\!\varcomppowercoef_\varcloud}\!) \!=\!\!
  \sum \limits_{\varslice \in \varslicesset}\!\varcost^{\varslice}\!(\vardecisionsvector,\varpolicy_\varbandwidthcoef,\varpolicy_{\!\varuplinkratecoef_{\varAP}}^{\varslice}\!,\varpolicy_{\!\varcomppowercoef_{\varcloud}}^{\varslice}\!) \!+\! \!\!\!\!\sum \limits_{\varplayer \in \varoffloaders_\varlocal(\vardecisionsvector)}\!\!\!\!\!\varcost_{\varplayer}^{\varlocal},\hspace{-0.63cm}
\end{equation}
where $(\varpolicy_{\varuplinkratecoef_\varAP},\varpolicy_{\varcomppowercoef_\varcloud}) = ((\varpolicy_{\varuplinkratecoef_\varAP}^{\varslice},\varpolicy_{\varcomppowercoef_\varcloud}^{\varslice}))_{\varslice \in \varslicesset}$ denotes the collection of slices' policies.

%\begin{equation}\label{eq::system_cost}
%  \rev{\varcost(\vardecisionsvector,\varpolicy_\varbandwidthcoef,\varpolicy_\varuplinkratecoef,\varpolicy_\varcomppowercoef) \!=\! \sum \nolimits_{\varplayer \in \varplayersset}
%  \varcost_{\varplayer}(\vardecisionsvector,\varpolicy_\varbandwidthcoef,\varpolicy_{\varuplinkratecoef_{\varAP}}^{\varslice},\varpolicy_\varcomppowercoef^\varslice).}
%\end{equation}
%\begin{equation}\label{eq::system_cost}
%  \varcost(\vardecisionsvector,\varpolicy_\varbandwidthcoef,\varpolicy_\varuplinkratecoef,\varpolicy_\varcomppowercoef) \!=\! \sum \nolimits_{\varslice \in \varslicesset}
%  \varcost^{\varslice}(\vardecisionsvector,\varpolicy_\varbandwidthcoef,\varpolicy_{\varuplinkratecoef_{\varAP}}^{\varslice},\varpolicy_\varcomppowercoef^\varslice).
%\end{equation}

\section{Problem Formulation}
\label{sec::_optimization_problem} 
We consider that the network operator aims at minimizing the system cost $\varcost(\vardecisionsvector,\varpolicy_{\varbandwidthcoef},\varpolicy_{\varuplinkratecoef_\varAP},\varpolicy_{\varcomppowercoef_\varcloud})$ by finding an optimal vector $\hat{\vardecisionsvector}$ of offloading decisions, and an optimal collection $(\hat{\varpolicy}_\varbandwidthcoef,\hat{\varpolicy}_{\varuplinkratecoef_{\varAP}},\hat{\varpolicy}_{\varcomppowercoef_{\varcloud}})$ of policies for sharing the edge resources across slices and within slices. 
We refer to the problem as the \emph{Joint Slice Selection and Edge Resource Management} (JSS-ERM) problem. 
Since the WDs generate atomic tasks that cannot be further split, the \ac{JSS-ERM} is a mixed-integer optimization problem, and can be formulated as 
\begin{eqnarray}
&\min \limits_{\vardecisionsvector,\varpolicy_{\varbandwidthcoef},\varpolicy_{\varuplinkratecoef_\varAP},\varpolicy_{\varcomppowercoef_\varcloud}}\varcost(\vardecisionsvector,\varpolicy_{\varbandwidthcoef},\varpolicy_{\varuplinkratecoef_\varAP},\varpolicy_{\varcomppowercoef_\varcloud})\label{obj::func} \\ 
& \textrm{s.t.}\sum \limits_{\vardecision \in \vardecisionsset_\varplayer} \varindicatorfunction_{ \{\vardecision_\varplayer = \vardecision \}} = 1, \forall \varplayer \in \varplayersset,  \label{cons::off_decisions_feasible} \\ 
&\varcost_\varplayer(\vardecisionsvector,\varpolicy_{\varbandwidthcoef},\varpolicy_{\varuplinkratecoef_{\varAP}},\varpolicy_{\varcomppowercoef_{\varcloud}}) \leq \vartimecost_\varplayer^{ex},  \forall \varplayer \in \varplayersset, \label{cons::delay_cons} \\
%\forall \vardecision \in \vardecisionsset_\varplayer,
&\sum \limits_{\varslice \in \varslicesset}\varbandwidthcoef_{\varAP}^\varslice \leq 1,  \forall \varAP \in \varAPsset, \label{cons::bandwidth_coeff_feasible} \\
%& \sum \limits_{\varplayer \in \varoffloaders_{(\varAP,\varslice)}(\vardecisionsvector)}\varuplinkratecoef_{\varplayer,\varAP}^\varslice \leq 1,  \forall \varAP \in \varAPsset, \forall \varslice \in \varslicesset, \label{cons::uplink_rate_coeff_feasible} \\
%& \sum \limits_{\varplayer \in \varoffloaders_{(\varcloud,\varslice)}(\vardecisionsvector)}\varcomppowercoef_{\varplayer,\varcloud}^\varslice \leq 1,  \forall \varcloud \in \varcloudsset, \forall \varslice \in \varslicesset, \label{cons::comp_power_coeff_feasible} \\
& \sum \limits_{\varoplayer \in \varoffloaders_{(\varedgeresource,\varslice)}(\vardecisionsvector)}\varuplinkratecoef_{\varoplayer,\varedgeresource}^\varslice \leq 1,  \forall \varedgeresource \in \varedgeresourcesset, \forall \varslice \in \varslicesset, \label{cons::in_slice_coeff_feasible} \\
%& \varoffvar_{\varplayer,\vardecision} \in \{0,1\}, \forall \varplayer \in \varplayersset, \forall \vardecision \in \vardecisionsset_\varplayer \label{cons::off_decisions_integral}, \\
&\varbandwidthcoef_\varAP^\varslice \geq 0, \forall \varAP \in \varAPsset, \forall \varslice \in \varslicesset, \label{cons::inter_coeff_continuous}\\
&\varuplinkratecoef_{\varplayer,\varedgeresource}^\varslice \geq 0, \forall \varplayer \in \varplayersset, \forall \varedgeresource \in \varedgeresourcesset, \forall \varslice \in \varslicesset. \label{cons::intra_coeff_continuous}
%&\varuplinkratecoef_{\varplayer,\varAP}^\varslice, \varcomppowercoef_{\varplayer,\varcloud}^\varslice \in [0,1], \forall \varplayer \in \varplayersset, \forall \varAP \in \varAPsset, \forall \varcloud \in \varcloudsset, \forall \varslice \in \varslicesset. \label{cons::intra_coeff_continuous}
\end{eqnarray}

Constraint~(\ref{cons::off_decisions_feasible}) enforces that each WD either performs the computation locally or offloads its task to exactly one logical resource $(\varAP,\varcloud,\varslice) \in \varAPsset \times \varcloudsset \times \varslicesset$; constraint~(\ref{cons::delay_cons}) ensures that the task completion time in the case of offloading is not greater than the task completion time in the case of local computing; constraint~(\ref{cons::bandwidth_coeff_feasible}) enforces a limitation on the amount of communication resources of an AP that can be provided to each slice; constraint~(\ref{cons::in_slice_coeff_feasible}) enforces a limitation on the amount of communication resources of an AP and the amount of computing resources of an EC that can be provided to each WD in each slice.
%constraint~(\ref{cons::uplink_rate_coeff_feasible}) enforces a limitation on the amount of communication resources of an AP that can be provided to each WD in each slice; %constraint~(\ref{cons::comp_power_coeff_feasible}) enforces a limitation on the amount of computing resources of an EC that can be provided to each WD in each slice; 

\begin{theorem}\label{theo::NP_hard_general}
	The	\ac{JSS-ERM} defined by (\ref{obj::func})-(\ref{cons::intra_coeff_continuous}) is NP-hard.
\end{theorem}

\begin{proof}
We provide the proof in Section~\ref{subsec::problem_complexity}.
%The proof follows from Theorem~\ref{theo::policies_offloading_decomposition} and Theorem~\ref{theo::NP_hard_offloading}.
\end{proof}

In what follows we develop an approximation scheme for the \ac{JSS-ERM} problem based on decomposition of the problem, and by adopting a game theoretic interpretation of one of the subproblems.%, based on which we also prove Theorem~\ref{theo::NP_hard_general}.}

%\section{\rev{Service Orchestration and Edge Resources Allocation}}
\section{Network Slice Orchestration and Edge Resource Allocation}
\label{sec::game}

In what follows we show that the \ac{JSS-ERM} problem can be solved through solving a series of smaller optimization problems. To do so, we start with considering the problem of finding the collection $(\hat{\varpolicy}_{\varbandwidthcoef},\hat{\varpolicy}_{\varuplinkratecoef_{\varAP}},\hat{\varpolicy}_{\varcomppowercoef_{\varcloud}})$ of optimal resource allocation policies for a given vector $\vardecisionsvector$ of offloading decisions.
%In this section we present how the \ac{JSS-ERM} problem can be decomposed without changing the optimal solution of (\ref{obj::func})-(\ref{cons::intra_coeff_continuous}).
\begin{lemma}\label{lemm::closed_form_general}
Consider an offloading decision vector $\vardecisionsvector$ for which the constraint~(\ref{cons::delay_cons}) can be satisfied. Furthermore, define the problem of finding a collection $(\hat{\varpolicy}_{\varbandwidthcoef},\hat{\varpolicy}_{\varuplinkratecoef_{\varAP}},\hat{\varpolicy}_{\varcomppowercoef_{\varcloud}})$  of optimal resource allocation policies as
\begin{eqnarray}
&\min \limits_{\varpolicy_{\varbandwidthcoef},\varpolicy_{\varuplinkratecoef_\varAP},\varpolicy_{\varcomppowercoef_\varcloud}}\varcost(\vardecisionsvector,\varpolicy_{\varbandwidthcoef},\varpolicy_{\varuplinkratecoef_\varAP},\varpolicy_{\varcomppowercoef_\varcloud})\label{obj::policies} \\ 
& \textrm{s.t.} (\ref{cons::delay_cons})-(\ref{cons::intra_coeff_continuous}).\label{cons::policies}
\end{eqnarray}
Then, the collection $(\hat{\varpolicy}_{\varbandwidthcoef},\hat{\varpolicy}_{\varuplinkratecoef_{\varAP}},\hat{\varpolicy}_{\varcomppowercoef_{\varcloud}})$ of optimal resource allocation policies sets the provisioning coefficients according to 
 \begin{eqnarray}
     &\hspace{-0.15cm}\hat{\varuplinkratecoef}_{\varplayer,\varedgeresource}^{\varslice} \!\!\!=\! \frac{\sqrt{\varedgeconstant_{\varplayer,\varedgeresource}^{\varslice}}}{\sum \limits_{j \in \varoffloaders_{(\varedgeresource,\varslice)}(\vardecisionsvector)}\sqrt{\varedgeconstant_{\varoplayer,\varedgeresource}^{\varslice}}},\!\forall \varedgeresource \!\in\! \varedgeresourcesset,\! \forall \varslice \!\in\! \varslicesset,\! \forall \varplayer \!\in\! \varoffloaders_{(\varedgeresource,\varslice)}\!(\vardecisionsvector) \label{eq::optimal_coef_slices}, \\
 	&\hat{\varbandwidthslicecoef}_\varAP^{\varslice} = \frac{\sum \limits_{j \in \varoffloaders_{(\varAP,\varslice)}(\vardecisionsvector)}\sqrt{\varedgeconstant_{\varoplayer,\varAP}^{\varslice}}}{\sum \limits_{\varoslice \in \varslicesset}\sum \limits_{\varoplayer \in \varoffloaders{(\varAP,\varoslice)}(\vardecisionsvector)}\sqrt{\varedgeconstant_{\varoplayer,\varAP}^{\varoslice}}}, \forall \varAP \in \varAPsset, \forall \varslice \in \varslicesset. \label{eq::optimal_coef_SRO}
	\end{eqnarray}
 \end{lemma}
\begin{proof}
First, observe that constraint~(\ref{cons::delay_cons}) can be omitted since we assumed that the decision vector $\vardecisionsvector$ is such that  constraint~(\ref{cons::delay_cons}) can be satisfied. Furthermore, by inspecting the leading minors of the Hessian matrix of the objective function~(\ref{obj::policies}) it is easy to show that the matrix is positive semidefinite on the domain defined by~(\ref{cons::policies}), and thus problem~(\ref{obj::policies})-(\ref{cons::policies}) is convex. Therefore, the optimal solution of the problem must satisfy the Karush–Kuhn–Tucker (KKT) conditions and thus we can formulate the corresponding Lagrangian dual problem. To do so, let us define $\varbandwidthcoefvector \triangleq (\varbandwidthcoef_\varAP^\varslice)_{\varslice \in \varslicesset,\varAP \in \varAPsset}$ and $\varuplinkratecoefvector^{\varslice} \triangleq (\varuplinkratecoef_{\varplayer,\varedgeresource}^{\varslice})_{\varplayer \in \varactiveplayersset, \varedgeresource \in \varedgeresourcesset}$, and let us introduce non-negative Lagrange multiplier vectors $\boldsymbol{\alpha} = (\alpha_{\varAP})_{\varAP \in \varAPsset}$, $\boldsymbol{\beta} = (\beta_{\varedgeresource}^{\varslice})_{\varedgeresource \in \varedgeresourcesset, \varslice \in \varslicesset}$, $\boldsymbol{\gamma} = (\gamma_{\varAP}^{\varslice})_{\varAP \in \varAPsset, \varslice \in \varslicesset}$ and $\boldsymbol{\delta} = (\delta_{\varplayer,\varedgeresource}^{\varslice})_{\varplayer \in \varoffloaders_{(\varedgeresource,\varslice)}(\vardecisionsvector),\varedgeresource \in \varedgeresourcesset, \varslice \in \varslicesset}$ for constraints in (\ref{cons::policies}), respectively. Next, let us define the Lagrangian dual problem corresponding to problem (\ref{obj::policies})-(\ref{cons::policies}) as $\max \limits_{\boldsymbol{\alpha}, \boldsymbol{\beta}, \boldsymbol{\gamma},\boldsymbol{\delta}\succeq 0} \min \limits_{\varbandwidthcoefvector,\varuplinkratecoefvector \succeq 0} \varlagrangian(\varbandwidthcoefvector, \varuplinkratecoefvector, \boldsymbol{\alpha}, \boldsymbol{\beta}, \boldsymbol{\gamma},\boldsymbol{\delta})$, where the Lagrangian is given by
\vspace{-0.3cm}
\begin{eqnarray}\label{eq::lagrangian_function_general}\nonumber
 &\hspace{-0.2cm}\varlagrangian(\varbandwidthcoefvector, \varuplinkratecoefvector,\boldsymbol{\alpha}, \boldsymbol{\beta}, \boldsymbol{\gamma},\boldsymbol{\delta}) \!=\! \sum \limits_{\varoslice \in \varslicesset} \sum \limits_{\varoedgeresource \in \varedgeresourcesset} \frac{1}{\varbandwidthcoef_{\varoedgeresource}^{\varoslice}} \big(\sum \limits_{j \in \varoffloaders_{(\varoedgeresource\!\!,\varoslice)}(\vardecisionsvector)}\frac{\varedgeconstant_{\varoplayer,\varoedgeresource}^{\varoslice}}{\varuplinkratecoef_{\varoplayer,\varoedgeresource}^{\varoslice}}\big)\!+\\\nonumber
  &\hspace{-0.2cm}\sum \limits_{\varoAP \in \varAPsset} \alpha_{\varoAP} (\sum \limits_{\varoslice \in \varslicesset}  \varbandwidthcoef_{\varoAP}^{\varoslice} \!-\! 1) \!+\!\! \sum \limits_{\varoedgeresource \in \varedgeresourcesset} \sum \limits_{\varoslice \in \varslicesset}\! \beta_{\varoedgeresource}^{\varoslice} \big(\sum \limits_{j \in \varoffloaders_{(\varoedgeresource\!\!,\varoslice)}(\vardecisionsvector)}\!\!\!\!\!\varuplinkratecoef_{\varoplayer,\varoedgeresource}^{\varoslice} \!-\! 1) \big)\\\nonumber
  &-\!\!\sum \limits_{\varoAP \in \varAPsset} \sum \limits_{\varoslice \in \varslicesset} \!\!\gamma_{\varoAP}^{\varoslice}\varbandwidthcoef_{\varoAP}^{\varoslice}- \!\!\sum \limits_{\varoedgeresource \in \varedgeresourcesset} \sum \limits_{\varoslice \in \varslicesset} \sum \limits_{j \in \varoffloaders_{(\varoedgeresource\!\!,\varoslice)}(\vardecisionsvector)} \!\!\!\!\delta_{\varoplayer,\varoedgeresource}^{\varoslice} \varuplinkratecoef_{\varoplayer,\varoedgeresource}^{\varoslice} + \!\!\!\!\sum \limits_{\varoplayer \in \varoffloaders_\varlocal(\vardecisionsvector)}\!\!\!\!\!\varcost_{\varoplayer}^{\varlocal}.\nonumber
 \end{eqnarray}
 Now, we can express the KKT conditions as follows
 \begin{eqnarray}
 \hspace{-1cm}\textrm{\scriptsize \textbf{stationarity:}} & \hspace{-0.4cm}\sum \limits_{\varoplayer \in \varoffloaders_{(\varAP,\varslice)}(\vardecisionsvector)}\!\!\!\frac{\varedgeconstant_{\varoplayer,\varAP}^{\varslice}}{\varuplinkratecoef_{\varoplayer,\varAP}^{\varslice}}\!\!\cdot\!\!\frac{1}{(\varbandwidthcoef_{\varAP}^{\varslice})^2} \!\!=\!\! \alpha_{\varAP} \!-\! \gamma_{\varAP}^{\varslice}, \varAP \!\in\! \varAPsset, \! \varslice \!\in\! \varslicesset,  \label{KKT::stationarity_SRO}\\
  &\hspace{-0.46cm}\frac{\varedgeconstant_{\varplayer,\varedgeresource}^{\varslice}}{\varbandwidthcoef_{\varedgeresource}^{\varslice}(\varuplinkratecoef_{\varplayer,\varedgeresource}^{\varslice})^2}\!\!=\!\! \beta_{\varedgeresource}^{\varslice} \!-\! \delta_{\varplayer,\varedgeresource}^{\varslice}, 
 \varedgeresource \!\!\in\!\! \varedgeresourcesset, \!\varslice \!\!\in\! \varslicesset, \!\varplayer \!\in\!\! \varoffloaders_{\!(\varedgeresource,\varslice)}(\vardecisionsvector), \label{KKT::stationarity_slices}\\
   \hspace{-0.3cm}\textrm{\scriptsize \textbf{pr. feasibility:}} &\hspace{-5.5cm} (\ref{cons::policies}), \label{KKT::primal_feasibility}\\
  \hspace{-0.1cm}\textrm{\scriptsize \textbf{du. feasibility:}} &\hspace{-3.9cm}\boldsymbol{\alpha}, \boldsymbol{\beta}, \boldsymbol{\gamma},\boldsymbol{\delta} \succeq 0, \label{KKT::dual_feasibility}\\
  \hspace{-0.4cm}\textrm{\scriptsize \textbf{co. slackness:}} &\hspace{-2.8cm}\alpha_{\varAP}(\sum \limits_{\varoslice \in \varslicesset}  \varbandwidthcoef_{\varAP}^{\varoslice} - 1), \varAP \!\in\! \varAPsset,  \label{KKT::complementary_slackness_1}\\
  \hspace{-1cm}\textrm{\scriptsize \textbf{\color{white}slackness:\color{black}}} &\hspace{-0.5cm}\beta_{\varedgeresource}^{\varslice}(\sum \limits_{\varoplayer \in \varoffloaders_{(\varedgeresource,\varslice)}(\vardecisionsvector)}\!\!\!\!\varuplinkratecoef_{\varoplayer,\varedgeresource}^{\varslice} - 1) = 0,\! \varedgeresource \!\in\! \varedgeresourcesset,\! \varslice \in \varslicesset, \label{KKT::complementary_slackness_2}\\
  \hspace{-1cm}\textrm{\scriptsize \textbf{\color{white} complementary \color{black}}} &\hspace{-2.7cm}-\gamma_\varAP^{\varslice}\varbandwidthslicecoef_{\varAP}^{\varslice} = 0,\! \varAP \!\in\! \varAPsset, \varslice \in \varslicesset  \label{KKT::complementary_slackness_3}\\
  \hspace{-1cm}\textrm{\scriptsize \textbf{\color{white} slackness: \color{black}}} &\hspace{-0.4cm}-\delta_{\varplayer,\varedgeresource}^{\varslice} \varuplinkratecoef_{\varplayer,\varedgeresource}^{\varslice} = 0, \varedgeresource \!\in\! \varedgeresourcesset, \!\varslice \!\in\! \varslicesset, \!\varplayer\in\! \varoffloaders_{(\varedgeresource,\varslice)}(\vardecisionsvector). \label{KKT::complementary_slackness_4}
  \end{eqnarray}
  We proceed with finding $\hat{\varuplinkratecoef}_{\varplayer,\varedgeresource}^{\varslice}$. First, from the KKT dual feasibility condition $\boldsymbol{\delta} \!\succeq\! 0$ and complementary slackness condition (\ref{KKT::complementary_slackness_4}) we obtain that $\delta_{\varplayer,\varedgeresource}^{\varslice} \!=\! 0$ must hold for every $\varedgeresource \!\in\! \varedgeresourcesset$, $\varslice \!\in\! \varslicesset$ and $i \!\in\! \varoffloaders_{(\varedgeresource,\varslice)}(\vardecisionsvector)$ as otherwise $\varuplinkratecoef_{\varplayer,\varedgeresource}^{\varslice} \!=\! 0$ would lead to infinite value of the objective function. Then, from the KKT stationarity condition (\ref{KKT::stationarity_slices}) and complementary slackness condition (\ref{KKT::complementary_slackness_2}) we obtain the expression~(\ref{eq::optimal_coef_slices}) for coefficients $\hat{\varuplinkratecoef}_{\varplayer,\varedgeresource}^{\varslice}$. Finally, by substituting expression~(\ref{eq::optimal_coef_slices}) into the KKT stationarity condition (\ref{KKT::stationarity_SRO}) and by following the same approach as for finding $\hat{\varuplinkratecoef}_{\varplayer,\varedgeresource}^{\varslice}$ we obtain the expression~(\ref{eq::optimal_coef_SRO}) for coefficients $\hat{\varbandwidthslicecoef}_{\varAP}^{\varslice}$, which proves the result.
\end{proof}

As a first step in the decomposition, let us consider the problem of finding the optimal collection $(\varpolicy_{\varuplinkratecoef_{\varAP}}^{*},\varpolicy_{\varcomppowercoef_{\varcloud}}^{*}) = ((\varpolicy_{\varuplinkratecoef_{\varAP}}^{\varslice,*},\varpolicy_{\varcomppowercoef_{\varcloud}}^{\varslice,*}))_{\varslice \in \varslicesset}$ of resource allocation policies of slices for a given vector $\vardecisionsvector$ of offloading decisions and a given policy $\varpolicy_\varbandwidthcoef$. % of the \ac{SRO}.
\begin{proposition}\label{lemm::closed_form_slices}
Consider an offloading decision vector $\vardecisionsvector$ for which  constraint~(\ref{cons::delay_cons}) can be satisfied and a policy $\varpolicy_\varbandwidthcoef$ for setting the inter-slice radio resource provisioning coefficients $\varbandwidthcoef_\varAP^\varslice, \forall \varAP \in \varAPsset, \forall \varslice \in \varslicesset$. Then the solution to the problem
\begin{align}
\min \limits_{\varpolicy_{\varuplinkratecoef_{\varAP}}^{\varslice},\varpolicy_{\varcomppowercoef_{\varcloud}}^\varslice}& \varcost^\varslice(\vardecisionsvector,\varpolicy_\varbandwidthcoef,\varpolicy_{\varuplinkratecoef_{\varAP}}^{\varslice},\varpolicy_{\varcomppowercoef_{\varcloud}}^\varslice) \label{eq::slices_BR}\\
&\text{s.t.} (\ref{cons::delay_cons}), (\ref{cons::in_slice_coeff_feasible}), (\ref{cons::intra_coeff_continuous}) \label{cons::slices_BR}.
\end{align}
%Furthermore, consider the problem~(\ref{eq::slices_BR})-(\ref{cons::slices_BR}) solved by slice $\varslice \in \varslicesset$.
%\begin{eqnarray}
%&\min \limits_{\varpolicy_{\varuplinkratecoef_\varAP},\varpolicy_{\varcomppowercoef_\varcloud}}\varcost(\vardecisionsvector,\varpolicy_{\varbandwidthcoef},\varpolicy_{\varuplinkratecoef_\varAP},\varpolicy_{\varcomppowercoef_\varcloud})\label{obj::intra_policies} \\ 
%& \textrm{s.t.} (\ref{cons::delay_cons}),(\ref{cons::in_slice_coeff_feasible}) \text{ and } (\ref{cons::intra_coeff_continuous}).\label{cons::intra_policies}
%\end{eqnarray}
%Then, equilibrium intra-slice resource allocation policies $(\varpolicy_{\varuplinkratecoef_{\varAP}}^{\varslice,*},\varpolicy_{\varcomppowercoef_{\varcloud}}^{\varslice,*})$ of slice $\varslice$ set the intra-slice provisioning coefficients according to
is given by ~(\ref{eq::optimal_coef_slices}), i.e., $(\varpolicy_{\varuplinkratecoef_{\varAP}}^{\varslice,*},\varpolicy_{\varcomppowercoef_{\varcloud}}^{\varslice,*}) = (\hat{\varpolicy}_{\varuplinkratecoef_{\varAP}}^\varslice,\hat{\varpolicy}_{\varcomppowercoef_{\varcloud}}^\varslice), \forall \varslice \in \varslicesset$.
%Consider an offloading decision vector $\vardecisionsvector$ for which the constraint~(\ref{cons::delay_cons}) can be satisfied and a policy $\varpolicy_\varbandwidthcoef$ according to which the \ac{SRO} sets inter-slice radio resource provisioning coefficients $\varbandwidthcoef_\varAP^\varslice, \forall \varAP \in \varAPsset, \forall \varslice \in \varslicesset$. Furthermore, define the problem of finding the collection $(\varpolicy_{\varuplinkratecoef_{\varAP}}^{*},\varpolicy_{\varcomppowercoef_{\varcloud}}^{*})$ of equilibrium intra-slice resource allocation policies as
%\begin{eqnarray}
%&\min \limits_{\varpolicy_{\varuplinkratecoef_\varAP},\varpolicy_{\varcomppowercoef_\varcloud}}\varcost(\vardecisionsvector,\varpolicy_{\varbandwidthcoef},\varpolicy_{\varuplinkratecoef_\varAP},\varpolicy_{\varcomppowercoef_\varcloud})\label{obj::intra_policies} \\ 
%& \textrm{s.t.} (\ref{cons::delay_cons}),(\ref{cons::in_slice_coeff_feasible}) \text{ and } (\ref{cons::intra_coeff_continuous}).\label{cons::intra_policies}
%\end{eqnarray}
%Then, the collection $(\varpolicy_{\varuplinkratecoef_{\varAP}}^{*},\varpolicy_{\varcomppowercoef_{\varcloud}}^{*})$ of equilibrium intra-slice resource allocation policies sets the intra-slice provisioning coefficients according to~(\ref{eq::optimal_coef_slices}).
\end{proposition}
\begin{proof}
	The result can be proved by following the approach presented in the proof of Lemma~\ref{lemm::closed_form_general}.
\end{proof}
As a second step, let us consider the problem of finding an optimal policy $\varpolicy_\varbandwidthcoef^*$ for a given vector $\vardecisionsvector$ of offloading decisions $\vardecisionsvector$ and the optimal collection $(\varpolicy_{\varuplinkratecoef_{\varAP}}^{*},\varpolicy_{\varcomppowercoef_{\varcloud}}^{*}) = (\hat{\varpolicy}_{\varuplinkratecoef_{\varAP}},\hat{\varpolicy}_{\varcomppowercoef_{\varcloud}})$ of the slices' policies.
\begin{proposition}
\label{lemm::closed_form_SRO}
	Consider an offloading decision vector $\vardecisionsvector$ for which the constraint~(\ref{cons::delay_cons}) can be satisfied. Furthermore, let us substitute (\ref{eq::optimal_coef_slices}) into (\ref{obj::func})-(\ref{cons::intra_coeff_continuous}) and define the problem of finding an optimal inter-slice radio resource allocation policy $\varpolicy_\varbandwidthcoef^*$, i.e., a solution to 
\begin{eqnarray}
&\min \limits_{\varpolicy_{\varbandwidthcoef}} \sum \limits_{\varoslice \in \varslicesset} \sum \limits_{\varoAP \in \varAPsset} \frac{1}{\varbandwidthcoef_{\varoAP}^{\varoslice}} \big(\sum \limits_{\varoplayer \in \varoffloaders_{(\varoAP\!\!,\varoslice)}(\vardecisionsvector)}\sqrt{\varedgeconstant_{\varoplayer,\varoAP}^{\varoslice}}\big)^2\label{obj::inter_policies}\\
&\textrm{s.t.} (\ref{cons::delay_cons}),(\ref{cons::bandwidth_coeff_feasible}) \text{ and } (\ref{cons::inter_coeff_continuous}). \label{cons::inter_policies}
\end{eqnarray}
	Then, the optimal inter-slice radio resource allocation policy $\varpolicy_\varbandwidthcoef^*$ sets the inter-slice provisioning coefficients according to~(\ref{eq::optimal_coef_SRO}), i.e., $\varpolicy_\varbandwidthcoef^* = \hat{\varpolicy}_{\varbandwidthcoef}$.
\end{proposition}
\begin{proof}
	The result can be proved by following the approach presented in the proof of Lemma~\ref{lemm::closed_form_general}.
\end{proof}

By combining the above two results, we are now ready to show that the \ac{JSS-ERM} problem can be decomposed into a sequence of optimization problems.
\begin{theorem}\label{theo::policies_decomposition}
The problem (\ref{obj::policies})-(\ref{cons::policies}) can be solved optimally by finding the optimal policies $(\hat{\varpolicy}_{\varuplinkratecoef_\varAP},\hat{\varpolicy}_{\varcomppowercoef_\varcloud})$ first, and finding the optimal policy $\hat{\varpolicy}_\varbandwidthcoef$ second, i.e.,
\begin{eqnarray}
&&\min \limits_{\varpolicy_{\varbandwidthcoef},\varpolicy_{\varuplinkratecoef_\varAP},\varpolicy_{\varcomppowercoef_\varcloud}}\varcost(\vardecisionsvector,\varpolicy_{\varbandwidthcoef},\varpolicy_{\varuplinkratecoef_\varAP},\varpolicy_{\varcomppowercoef_\varcloud})=\nonumber\\
&&\min \limits_{\varpolicy_{\varbandwidthcoef}}\min \limits_{\varpolicy_{\varuplinkratecoef_\varAP},\varpolicy_{\varcomppowercoef_\varcloud}}\varcost(\vardecisionsvector,\varpolicy_{\varbandwidthcoef},\varpolicy_{\varuplinkratecoef_\varAP},\varpolicy_{\varcomppowercoef_\varcloud})\nonumber\\ 
\end{eqnarray}
\end{theorem}
\begin{proof}
The result follows from the proofs of Lemma~\ref{lemm::closed_form_general}, Proposition~\ref{lemm::closed_form_slices} and Proposition~\ref{lemm::closed_form_SRO}. 
\end{proof}
Furthermore, as the next theorem shows, we can use this decomposition structure also for computing the optimal offloading decision vector.
\begin{theorem}\label{theo::policies_offloading_decomposition}
The problem (\ref{obj::func})-(\ref{cons::intra_coeff_continuous}) can be solved optimally by finding the optimal collection $(\hat{\varpolicy}_\varbandwidthcoef,\hat{\varpolicy}_{\varuplinkratecoef_\varAP},\hat{\varpolicy}_{\varcomppowercoef_\varcloud})$ of resource allocation policies first, and finding an optimal offloading decision vector $\hat{\vardecisionsvector}$ second, i.e.,
\begin{eqnarray}\label{eq::finding_opt_offloading_decisions}
&&\min \limits_{\vardecisionsvector, \varpolicy_{\varbandwidthcoef},\varpolicy_{\varuplinkratecoef_\varAP},\varpolicy_{\varcomppowercoef_\varcloud}}\varcost(\vardecisionsvector,\varpolicy_{\varbandwidthcoef},\varpolicy_{\varuplinkratecoef_\varAP},\varpolicy_{\varcomppowercoef_\varcloud})=\nonumber\\
&&\min \limits_{\vardecisionsvector} \min \limits_{\varpolicy_{\varbandwidthcoef}}\min \limits_{\varpolicy_{\varuplinkratecoef_\varAP},\varpolicy_{\varcomppowercoef_\varcloud}}\varcost(\vardecisionsvector,\varpolicy_{\varbandwidthcoef},\varpolicy_{\varuplinkratecoef_\varAP},\varpolicy_{\varcomppowercoef_\varcloud})\nonumber\\ 
\end{eqnarray}
\end{theorem}
\begin{proof}
It is easy to see that the exact values of the provisioning coefficients are functions of $\hat{\vardecisionsvector}$. 
However, the optimal policies according to which the resources are shared are the same for every offloading decision vector $\vardecisionsvector \in \vardecisionsset$, as defined by (\ref{eq::optimal_coef_slices}) and (\ref{eq::optimal_coef_SRO}). Therefore, one can solve the problem~(\ref{obj::policies})-(\ref{cons::policies}) first, assuming an arbitrary offloading decision vector $\vardecisionsvector$, and then given the solution $(\hat{\varpolicy}_\varbandwidthcoef,\hat{\varpolicy}_{\varuplinkratecoef_\varAP},\hat{\varpolicy}_{\varcomppowercoef_\varcloud})$ of~(\ref{obj::policies})-(\ref{cons::policies}) find the optimal offloading decision vector $\hat{\vardecisionsvector}$ that will determine the exact values of the provisioning coefficients. This proves the result.  
\end{proof}

\begin{comment}
\begin{corollary}\label{cor::optimal_decomposition}
The \ac{NSO-ERA} game preserves the optimality of the \ac{JSS-ERM} problem.
\end{corollary}
\begin{proof}
Since the \ac{NSO-ERA} game is solved using backward induction i.e., we first consider problem (\ref{eq::slices_BR})-(\ref{cons::slices_BR}) solved by the slices, and then we consider problem (\ref{eq::operator_BR_complete_information})-(\ref{cons::operator_BR_complete_information}) solved by the \ac{SRO}, the result follows from Theorem~\ref{theo::policies_decomposition} and Theorem~\ref{theo::policies_offloading_decomposition}.
\end{proof}
\end{comment}

\subsection{Discussion and Practical Implications}
\begin{figure}[t]
	%\begin{minipage}{0.5\textwidth}		
	\begin{center}
		\includegraphics[width=0.8\columnwidth]{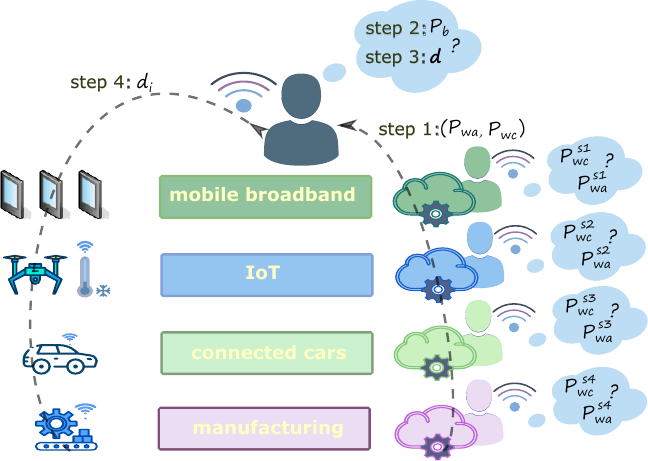}
			\caption{An example of the potential implementation of a resource allocation and orchestration framework.}
		\label{fig::implementation}
	\end{center}
	%\end{minipage}
	\vspace{-0.2cm}
\end{figure}
So far we have shown that the \ac{JSS-ERM} problem can be decomposed into a  $S+2$ coupled resource allocation problems that can be solved sequentially. It is of interest to discuss the relationship between the decomposition and the potential implementation of a resource allocation and orchestration framework. 

The proposed decomposition results in an optimization problem to be solved at the network level (eqns. (\ref{obj::inter_policies})-(\ref{cons::inter_policies})) and one in each slice ((eqns. (\ref{eq::slices_BR})-(\ref{cons::slices_BR}))), followed by the problem of finding an optimal offloading decision vector.
This structure is aligned with the slice-based network architecture proposed in~\cite{461563398f1a4f8c876135ad7e7d5bed}, where inter-slice radio resource allocation and service orchestration are performed by a centralized entity, the \ac{SRO}, while intra-slice radio and computing resource management is performed by the slices themselves, i.e., each slice manages its own radio and computing resources.

Figure~\ref{fig::implementation} illustrates the interaction between the \ac{SRO} and slices in the potential implementation of a resource allocation and orchestration framework.

\subsection{Problem Complexity}
\label{subsec::problem_complexity}
In what follows we provide a result concerning the complexity of the \ac{JSS-ERM} problem. 
For notational convenience let us first define the set of resources $\Tilde{\varresourcesset} \triangleq \{\{\varAPsset \times \varslicesset\} \cup \{\varcloudsset \times \varslicesset\} \cup \varactiveplayersset\}$ and let us introduce the following shorthand notation
\begin{eqnarray}\label{eq::weights}\nonumber
  &\hspace{-0.2cm}\varcongestionweight_{\varplayer,(\varAP,\varslice)} \triangleq \sqrt{\frac{\vardatasize_\varplayer}{\varphywirelessrate_{\varplayer,\varAP}}},  \text{        }\varcongestionweight_{\varplayer,(\varcloud,\varslice)}\triangleq \sqrt{\rev{\frac{\vartaskcomplexity_\varplayer}{\vardeviceslicecoef_{\varplayer,\varslice}}}}, \text{        }\varcongestionweight_{\varplayer,\varplayer} \triangleq \sqrt{\vartaskcomplexity_\varplayer},\\
%\end{equation}
%\begin{equation}\label{eq::congestions}
   &\varcongestionweight_{\varresource}(\vardecisionsvector) \triangleq \sum \limits_{\varoplayer \in \varoffloaders_{\varresource}(\vardecisionsvector)}\varcongestionweight_{\varoplayer,\varresource}, \forall \varresource \in \Tilde{\varresourcesset}.
\end{eqnarray}

First, by substituting (\ref{eq::optimal_coef_slices}) into (\ref{eq::user_i_cost}) and by using the notation introduced in (\ref{eq::weights}), we can express the cost of WD $\varplayer$ under a policy $\varpolicy_\varbandwidthcoef$ and the collection $(\hat{\varpolicy}_{\varuplinkratecoef_\varAP},\hat{\varpolicy}_{\varcomppowercoef_\varcloud})$ of optimal allocation policies of slices as 
\begin{equation}\label{eq::WD_cost_under_opt_slices}
    \Tilde{\varcost}_\varplayer(\vardecisionsvector) = \sum \limits_{\varresource \in \Tilde{\varresourcesset}_{\vardecision_\varplayer}} m_\varresource \varcongestionweight_{\varplayer,\varresource}\varcongestionweight_\varresource(\vardecisionsvector),
\end{equation}
where $\Tilde{\varresourcesset}_{\vardecision_\varplayer}$ is the set of resources that WD $\varplayer$ uses for performing its task in $\vardecisionsvector$ (i.e., $\Tilde{\varresourcesset}_{\vardecision_\varplayer} \subset \Tilde{\varresourcesset}$) and $m_{(\varAP,\varslice)} = 1/\varbandwidthcoef_\varAP^\varslice$, $m_{(\varcloud,\varslice)} = 1/\varcomppower_\varcloud^\varslice$ and $m_{\varplayer} = 1/\varcomppower_\varplayer^\varlocal$.

Second, by summing the expressions (\ref{eq::WD_cost_under_opt_slices}) over all WDs $\varplayer \in \varactiveplayersset$ and by reordering the summations we can express the system cost~(\ref{eq::system_cost}) under a policy $\varpolicy_\varbandwidthcoef$ and the collection $(\hat{\varpolicy}_{\varuplinkratecoef_\varAP},\hat{\varpolicy}_{\varcomppowercoef_\varcloud})$ of optimal allocation policies of slices as 
\begin{eqnarray}\label{eq::system_cost_under_optimal_slices}
    &\hspace{-0.6cm}\Tilde{\varcost}(\vardecisionsvector) = \sum \limits_{\varresource \in \Tilde{\varresourcesset}} m_\varresource\varcongestionweight_\varresource^2(\vardecisionsvector).
\end{eqnarray}

Next, let us define the set of resources $\bar{\varresourcesset} \triangleq \{\varAPsset \cup \{\varcloudsset \times \varslicesset\} \cup \varactiveplayersset\}$ and a coefficient $\varcongestionweight_{\varplayer,\varAP} \triangleq \varcongestionweight_{\varplayer,(\varAP,\varslice)} = \sqrt{\vardatasize_\varplayer/\varphywirelessrate_{\varplayer,\varAP}}$. By substituting (\ref{eq::optimal_coef_SRO}) into (\ref{eq::WD_cost_under_opt_slices}), we can express the cost of WD $\varplayer$ under the collection $(\hat{\varpolicy}_\varbandwidthcoef,\hat{\varpolicy}_{\varuplinkratecoef_\varAP},\hat{\varpolicy}_{\varcomppowercoef_\varcloud})$ of optimal allocation policies as 
\begin{equation}\label{eq::WD_cost_under_opt_slices_opt_operato}
    \bar{\varcost}_\varplayer(\vardecisionsvector) = \sum \limits_{\varresource \in \bar{\varresourcesset}_{\vardecision_\varplayer}} m_\varresource \varcongestionweight_{\varplayer,\varresource}\varcongestionweight_\varresource(\vardecisionsvector),
\end{equation}
where $\bar{\varresourcesset}_{\vardecision_\varplayer}$ is the set of resources that WD $\varplayer$ uses for performing its task in $\vardecisionsvector$ (i.e., $\bar{\varresourcesset}_{\vardecision_\varplayer} \subset \bar{\varresourcesset}$) and $m_{\varAP} = 1$.

Finally, by summing the expressions (\ref{eq::WD_cost_under_opt_slices_opt_operato}) over all WDs $\varplayer \in \varactiveplayersset$ and by reordering the summations we can express the system cost (\ref{eq::system_cost_under_optimal_slices}) under the collection $(\hat{\varpolicy}_\varbandwidthcoef,\hat{\varpolicy}_{\varuplinkratecoef_\varAP},\hat{\varpolicy}_{\varcomppowercoef_\varcloud})$ of optimal allocation policies as 
\begin{eqnarray}\label{eq::system_cost_under_optimal_slices_optimal_operator}
    &\hspace{-0.6cm}\bar{\varcost}(\vardecisionsvector) = \sum \limits_{\varresource \in \bar{\varresourcesset}} m_\varresource\varcongestionweight_\varresource^2(\vardecisionsvector).
\end{eqnarray}
\color{black}
\begin{theorem}\label{theo::NP_hard_offloading}
    Consider the problem of finding the optimal vector $\hat{\vardecisionsvector}$ of offloading decisions of WDs under the collection $(\hat{\varpolicy}_\varbandwidthcoef,\hat{\varpolicy}_{\varuplinkratecoef_\varAP},\hat{\varpolicy}_{\varcomppowercoef_\varcloud})$ of optimal allocation policies that set provisioning coefficients according to (\ref{eq::optimal_coef_slices}) and (\ref{eq::optimal_coef_SRO})
    \begin{eqnarray}
&\min \limits_{\vardecisionsvector} \bar{\varcost}(\vardecisionsvector)\label{obj::offloading}\\
&\textrm{s.t.} (\ref{cons::off_decisions_feasible}). \label{cons::offloading}
\end{eqnarray}
	Problem (\ref{obj::offloading})-(\ref{cons::offloading}) is NP-hard.
\end{theorem}

\begin{proof}
	\rev{We prove the NP-hardness of the problem by reduction from the \emph{Minimum Sum of Squares} problem (SP19 problem in~\cite{garey2002computers}):  given a finite set $\mathcal{B}$, a size $s(b) \in \mathbb{Z}^{+}, \forall b \in \mathcal{B}$ and positive integers $K \leq |\mathcal{B}|$ and $J$, the question is whether $\mathcal{B}$ can be partitioned into $K$ disjoint subsets $\mathcal{B}_1,\mathcal{B}_2,\hdots,\mathcal{B}_K$ such that $\sum \limits_{k = 1}^{K} \Big(\sum \limits_{b \in \mathcal{B}_k}\!\!s(b)\Big)^2 \!\!\leq\!\! J$. }
	
	For the reduction we set $\varslicessetdym = 1$, $\varcloudssetdym = 0$ and $\varcomppower_\varplayer^\varlocal = 0$, $\forall \varplayer \in \varplayersset$, i.e., in this simplified version of the problem $\Tilde{\varresourcesset} = \varAPsset$. Next, we let $\varplayersset = \mathcal{B}$, $|\varAPsset| = K$, $\varphywirelessrate_{\varplayer,\varAP} = \varphywirelessrate_\varplayer$, $\forall \varplayer \in \varplayersset$, $\forall \varAP \in \varAPsset$ and $\sqrt{\vardatasize_\varplayer/\varphywirelessrate_\varplayer} = s(b)$. Then, it follows from (\ref{eq::system_cost_under_optimal_slices}) that the optimal solution of (\ref{obj::offloading})-(\ref{cons::offloading}) provides the solution to the SP19 problem. As SP19 is NP-hard, problem (\ref{obj::offloading})-(\ref{cons::offloading}) is also NP-hard, which proves the theorem.
\end{proof}
\begin{proof}[Proof of Theorem~\ref{theo::NP_hard_general}]
The result follows from Theorem~\ref{theo::policies_offloading_decomposition} and Theorem~\ref{theo::NP_hard_offloading}. 
\end{proof}
\color{black}
\section{Approximation Scheme for the JSS-ERM Problem}
\label{sec::slice_level}
In what follows we propose the \ac{COS} algorithm for computing an approximation to the optimal solution of the \ac{JSS-ERM} problem. In particular, the algorithm serves as an approximation scheme to the problem of finding an optimal offloading decision vector. The algorithm starts from an offloading decision vector $\vardecisionsvector^0$ in which all WDs perform  computation locally and it lets WDs update their offloading decisions one at a time, based on their local cost function $\tilde{\varcost}_\varplayer(\vardecisionsvector)$. We show the pseudo code of the algorithm in Figure~\ref{fig::COS}.
\begin{figure}[t]
	{\ruleline{$\vardecisionsvector{}^{*} = \ac{COS}(\vardecisionsvector^0,\varpolicy_\varbandwidthcoef,\varpolicy_\varuplinkratecoef^*,\varpolicy_\varcomppowercoef^*)$}
	 %\vspace{-0.9cm}
	\begin{algorithm}[H]
	$\vardecisionsvector \leftarrow \vardecisionsvector^0$\\
		\While{\!$\exists$\!\! WD \!$\varoplayer \!\in\! \varactiveplayersset$ s.t. $\vardecision_\varoplayer \neq \argmin \limits_{\vardecision{}^{\prime}_\varoplayer \in \varactivedecisionsset_\varoplayer}\! \Tilde{\varcost}_\varoplayer(\vardecision{}^{\prime}_\varoplayer,\! \vardecision_{-\varoplayer})$}
		 {
		 $\vardecision{}^{*}_\varoplayer=\argmin \limits_{\vardecision{}^{\prime}_\varoplayer \in \varactivedecisionsset_\varoplayer} \Tilde{\varcost}_\varoplayer(\vardecision{}^{\prime}_\varoplayer, \vardecision_{-\varoplayer})$,
		 $\vardecisionsvector = (\vardecision{}^{*}_\varoplayer, \vardecision_{-\varoplayer})$
 	    }
 	    $\vardecisionsvector{}^{*} = \vardecisionsvector$
	\end{algorithm}
	%\vspace{-0.4cm}
	\hrule
	\caption{Pseudo code of the \ac{COS} algorithm.}\label{fig::COS}
	\vspace{-0.2cm}
	}
\end{figure}
\begin{theorem}\label{theo::COS_convergence}
The \ac{COS} algorithm terminates after a finite number of the iterations for any allocation policy $\varpolicy_\varbandwidthcoef$ and the collection $(\hat{\varpolicy}_{\varuplinkratecoef_{\varAP}},\hat{\varpolicy}_{\varuplinkratecoef_{\varcloud}})$ of optimal allocation policies of slices.
\end{theorem}
\begin{proof}
\revo{The proof is based on a game theoretic treatment of the problem 
\begin{eqnarray}
&\min \limits_{\vardecisionsvector} \tilde{\varcost}(\vardecisionsvector)\label{obj::offloading_and_Pb}\\
&\textrm{s.t.} (\ref{cons::off_decisions_feasible}), \label{cons::offloading_and_Pb}
\end{eqnarray}
in which the inter-slice radio resource provisioning coefficients are set according to an arbitrary policy $\varpolicy_\varbandwidthcoef$ and the intra-slice radio and computing power provisioning coefficients are set according to the optimal policies $\hat{\varpolicy}_{\varuplinkratecoef_{\varAP}}$ and $\hat{\varpolicy}_{\varuplinkratecoef_{\varcloud}}$, respectively.}

In what follows we show that the problem (\ref{obj::offloading_and_Pb})-(\ref{cons::offloading_and_Pb}) can be interpreted as a congestion game $\varthegame(\varpolicy_\varbandwidthcoef,\hat{\varpolicy}_{\varuplinkratecoef_{\varAP}},\hat{\varpolicy}_{\varcomppowercoef_{\varcloud}}) = <\varactiveplayersset,(\varactivedecisionsset_\varplayer)_{\varplayer \in \varactiveplayersset},(\Tilde{\varcost}_\varplayer)_{\varplayer \in \varactiveplayersset}>$ with resource-dependent weights $\varcongestionweight_{\varplayer,\varresource}$, $\varplayer \in \varactiveplayersset$, $\varresource \in \Tilde{\varresourcesset}$, and the cost of WD $\varplayer$ in the resulting game is given by~(\ref{eq::WD_cost_under_opt_slices}).
First, observe that $\varcongestionweight_{\varplayer,\varresource}$ can be interpreted as the weight that WD $\varplayer$ contributes to the congestion when using resource $\varresource \in \Tilde{\varresourcesset}$ and thus $\varcongestionweight_{\varresource}(\vardecisionsvector)$ can be interpreted as the total congestion on resource $\varresource$ in strategy profile $\vardecisionsvector$. This in fact implies that the cost (\ref{eq::WD_cost_under_opt_slices}) of WD $\varplayer$ in strategy profile $\vardecisionsvector$ depends on its own resource-dependent weights $\varcongestionweight_{\varplayer,\varresource}$ and on the total congestion $\varcongestionweight_{\varresource}(\vardecisionsvector)$ on the resources it uses. \revo{Therefore, it follows from~\cite{harks2011characterizing} that the problem (\ref{obj::offloading_and_Pb})-(\ref{cons::offloading_and_Pb}) can be interpreted as a congestion game $\varthegame(\varpolicy_\varbandwidthcoef,\hat{\varpolicy}_{\varuplinkratecoef_{\varAP}},\hat{\varpolicy}_{\varcomppowercoef_{\varcloud}})$ with resource dependent weights. 
\revo{Consequently, the \ac{COS} algorithm terminates after a finite number of iterations iff the game $\varthegame(\varpolicy_\varbandwidthcoef,\hat{\varpolicy}_{\varuplinkratecoef_{\varAP}},\hat{\varpolicy}_{\varcomppowercoef_{\varcloud}})$ has a pure strategy Nash equilibrium. \footnote{A pure strategy Nash equilibrium of a strategic game is a collection $\vardecisionsvector^*$ of decisions (called a strategy profile) for which $\tilde{\varcost}_\varplayer(\vardecision^*_\varplayer, \vardecision^*_{-\varplayer})\leq \tilde{\varcost}_\varplayer(\vardecision_\varplayer, \vardecision^*_{-\varplayer})$, $\forall \vardecision_\varplayer$, where $\vardecision^*_{\varplayer}$ and $\vardecision^*_{-\varplayer}$ are standard game theoretical notations for an improvement step of player $\varplayer$ and for the collection of decisions (strategies) of all players other than $\varplayer$, respectively.}}}

Since the cost $\varresourcecost_{\varresource}(\vardecisionsvector) \triangleq m_{\varresource}\varcongestionweight_{\varresource}(\vardecisionsvector)$ of sharing every resource $\varresource \in \Tilde{\varresourcesset}$ is an affine function of the congestion $\varcongestionweight_{\varresource}(\vardecisionsvector)$ on resource $\varresource$, it follows from Theorem $4.2$ in~\cite{harks2011characterizing} that the game $\varthegame(\varpolicy_\varbandwidthcoef,\hat{\varpolicy}_{\varuplinkratecoef_{\varAP}},\hat{\varpolicy}_{\varcomppowercoef_{\varcloud}})$ has the exact potential function\footnote{A function $\varpotentialfunction: \times_\varplayer (\varactivedecisionsset_\varplayer) \rightarrow \mathbb{R}$ is an exact potential for a finite strategic game if for an arbitrary strategy profile $(\vardecision_\varplayer,\vardecision_{-\varplayer})$ and for any improvement step $\vardecision{}^*_\varplayer$ the following holds:
\begin{equation}\label{eq::potental_def}
     \varpotentialfunction(\vardecision_\varplayer,\vardecision_{-\varplayer})\!-\!\varpotentialfunction(\vardecision{}^*_\varplayer,\vardecision_{-\varplayer}) \!=\!   \Tilde{\varcost}_\varplayer(\vardecision_\varplayer,\vardecision_{-\varplayer}) \!-\! \Tilde{\varcost}_\varplayer(\vardecision{}^*_\varplayer,\vardecision_{-\varplayer}). 
\end{equation}
} given by 
\begin{equation}\label{eq::potential_function_1}
  \varpotentialfunction(\vardecisionsvector) = \sum \limits_{\varplayer \in \varactiveplayersset} \sum \limits_{\varresource \in \tilde{\varresourcesset}_{\vardecision_\varplayer}}\varcongestionweight_{\varoplayer,\varresource}\varresourcecost_\varresource^{\leq \varplayer}(\vardecisionsvector),  
\end{equation}
where $\varresourcecost_\varresource^{\leq \varplayer}(\vardecisionsvector) = m_{\varresource}\varcongestionweight_{\varresource}^{\leq \varplayer}(\vardecisionsvector)$ and $\varcongestionweight_{\varresource}^{\leq \varplayer}(\vardecisionsvector) = \sum \limits_{\{\varoplayer \in \varoffloaders_\varresource(\vardecisionsvector)| \varoplayer \leq \varplayer\}}\varcongestionweight_{\varplayer,\varresource}$.

It is well known that in a finite strategic game that admits an exact potential all improvement paths\footnote{An improvement path is a sequence of strategy profiles in which one player at a time changes its strategy through performing an improvement step.} are finite~\cite{monderer1996potential} and thus the existence of the exact potential function~(\ref{eq::potential_function_1}) allows us to use the \ac{COS} algorithm for computing a pure strategy Nash equilibrium $\vardecisionsvector^*$ of the game $\varthegame(\varpolicy_\varbandwidthcoef,\hat{\varpolicy}_{\varuplinkratecoef_{\varAP}},\hat{\varpolicy}_{\varcomppowercoef_{\varcloud}})$, which proves the result.
\end{proof}

\begin{theorem}\label{theo::COS_convergence_all_optimal}
The \ac{COS} algorithm terminates after a finite number of the iterations for the collection $(\hat{\varpolicy}_{\varbandwidthslicecoef},\hat{\varpolicy}_{\varuplinkratecoef_{\varAP}},\hat{\varpolicy}_{\varuplinkratecoef_{\varcloud}})$ of optimal allocation policies.
\end{theorem}
\begin{proof}
By following the same approach as in the proof of Theorem~\ref{theo::COS_convergence}, it is easy to show that given the collection $(\hat{\varpolicy}_{\varbandwidthslicecoef},\hat{\varpolicy}_{\varuplinkratecoef_{\varAP}},\hat{\varpolicy}_{\varuplinkratecoef_{\varcloud}})$ of optimal allocation policies, the problem~(\ref{obj::offloading})-(\ref{cons::offloading}) can be interpreted as a congestion game $\varthegame(\hat{\varpolicy}_\varbandwidthcoef,\hat{\varpolicy}_{\varuplinkratecoef_{\varAP}},\hat{\varpolicy}_{\varcomppowercoef_{\varcloud}}) = <\varactiveplayersset,(\varactivedecisionsset_\varplayer)_{\varplayer \in \varactiveplayersset},(\bar{\varcost}_\varplayer)_{\varplayer \in \varactiveplayersset}>$ with resource-dependent weights $\varcongestionweight_{\varplayer,\varresource}$, $\varplayer \in \varactiveplayersset$, $\varresource \in \bar{\varresourcesset}$, and the cost of WD $\varplayer$ in the resulting game is given by~(\ref{eq::WD_cost_under_opt_slices_opt_operato}). 

\revo{Since $m_\varAP = 1, \forall \varAP \in \varAPsset$, the cost $\varresourcecost_{\varresource}(\vardecisionsvector) \triangleq m_{\varresource}\varcongestionweight_{\varresource}(\vardecisionsvector)$ of sharing every resource $\varresource \in \bar{\varresourcesset}$ is an affine function of the congestion on resource $\varresource$. 
%, it follows from Theorem $4.2$ in~\cite{harks2011characterizing} that the game $\varthegame(\hat{\varpolicy}_\varbandwidthcoef,\hat{\varpolicy}_{\varuplinkratecoef_{\varAP}},\hat{\varpolicy}_{\varcomppowercoef_{\varcloud}})$ has the exact potential function given by 
%\begin{equation}\label{eq::potential_function_2}
%    \varpotentialfunction(\vardecisionsvector) = \sum \limits_{\varplayer \in \varactiveplayersset} \sum \limits_{\varresource \in \bar{\varresourcesset}_{\vardecision_\varplayer}}\varcongestionweight_{\varplayer,\varresource}\varresourcecost_\varresource^{\leq \varplayer}(\vardecisionsvector).
%\end{equation}
Therefore, the game $\varthegame(\hat{\varpolicy}_\varbandwidthcoef,\hat{\varpolicy}_{\varuplinkratecoef_{\varAP}},\hat{\varpolicy}_{\varcomppowercoef_{\varcloud}})$ is also an exact potential game, and thus the \ac{COS} algorithm computes a pure strategy Nash equilibrium $\vardecisionsvector^*$ of the game $\varthegame(\hat{\varpolicy}_\varbandwidthcoef,\hat{\varpolicy}_{\varuplinkratecoef_{\varAP}},\hat{\varpolicy}_{\varcomppowercoef_{\varcloud}})$, which proves the result.}
\end{proof}
%\color{blue}

In general, the number of improvement steps can be exponential in a potential game, but as we show next the \ac{COS} algorithm can compute an equilibrium $\vardecisionsvector^*$ of offloading decisions efficiently. 

\begin{theorem}
The \ac{COS} algorithm terminates in $\mathcal{O}(n\frac{\varcost^{min}}{\varcost^{max}}\log \frac{\sum_{\varplayer \in \varplayersset}\vartime_\varplayer^{ex}}{\varpotentialfunction^{min}})$ iterations, where $n \geq 1$, $\varcost^{min}$ and $\varcost^{max}$ are system parameter dependent constants and $\varpotentialfunction^{min}$ is the minimum value of the potential function.
\end{theorem}
\begin{proof}
First, let us define the minimum cost that WD $\varplayer$ can achieve as $\varcost_\varplayer^{min} \triangleq \min \{\varcost_\varplayer^\varlocal, \min_{(\varAP,\varcloud,\varslice) \in \varAPsset \times \varcloudsset \times \varslicesset} (\vardatasize_\varplayer/\varphywirelessrate_{\varplayer,\varAP} + \vartaskcomplexity_{\varplayer,\varslice}/\varcomppower_\varcloud^\varslice)\}$ and let $\varcost^{min} \triangleq \min_{\varplayer \in \varplayersset}\varcost_\varplayer^{min}$. Furthermore, let us define the maximum cost that WD $\varplayer$ can achieve if it was the only WD in the system as $\varcost_\varplayer^{max} \triangleq \max \{\varcost_\varplayer^\varlocal, \min_{(\varAP,\varcloud,\varslice) \in \varAPsset \times \varcloudsset \times \varslicesset} (\vardatasize_\varplayer/\varphywirelessrate_{\varplayer,\varAP} + \vartaskcomplexity_{\varplayer,\varslice}/\varcomppower_\varcloud^\varslice)\}$, and let $\varcost^{max} = \max_{\varplayer \in \varplayersset}\varcost_\varplayer^{max}$.   

Consider now an iteration of the \ac{COS} algorithm where the offloading decision of WD $\varplayer$ is updated from $\vardecision_\varplayer$ to $\vardecision_\varplayer^*$.  We can then write
\begin{eqnarray}\label{eg::potential_decreas_bound}
   &\varpotentialfunction(\vardecision_\varplayer,\vardecision_{-\varplayer})\!-\!\varpotentialfunction(\vardecision{}^*_\varplayer,\vardecision_{-\varplayer}) \!=\! \Tilde{\varcost}_\varplayer(\vardecision_\varplayer,\vardecision_{-\varplayer}) \!-\! \Tilde{\varcost}_\varplayer(\vardecision{}^*_\varplayer,\vardecision_{-\varplayer})\nonumber \\
   &\geq -\varcost^{max} \geq -\frac{\varcost^{max}}{\varcost^{min}}\varpotentialfunction(\vardecision_\varplayer,\vardecision_{-\varplayer}),
\end{eqnarray}
where the equality follows from the definition of the exact potential function~(\ref{eq::potental_def}), the first inequality follows from the fact that $\Tilde{\varcost}_\varplayer(\vardecision_\varplayer,\vardecision_{-\varplayer}) \!-\! \Tilde{\varcost}_\varplayer(\vardecision{}^*_\varplayer,\vardecision_{-\varplayer}) > 0$ since $\vardecision_\varplayer^*$ is an improvement step of WD $\varplayer$ and the last inequality follows from the fact that $\varpotentialfunction(\vardecision_\varplayer,\vardecision_{-\varplayer}) \geq \varcost^{min}$ for any vector $\vardecisionsvector$ of offloading decisions.

Therefore, from~(\ref{eg::potential_decreas_bound}) we obtain $\varpotentialfunction(\vardecision_\varplayer^*,\vardecision_{-\varplayer}) \leq (1 + \frac{\varcost^{max}}{{\varcost^{min}}})$, i.e., the \ac{COS} algorithm decreases the potential function by at least a factor of $(1+\frac{\varcost^{max}}{\varcost^{min}})$. Next, observe that from the definition of the constants $\varcost^{max}$ and $\varcost^{min}$ we have $\frac{\varcost^{max}}{{\varcost^{min}}} \geq 1$. Hence, since $(1 + x)^{\frac{n}{x}} \leq e^n$ holds for $x, n \geq 1$, we obtain that after every $n\frac{\varcost^{min}}{\varcost^{max}}$ iterations of the \ac{COS} algorithm $(1 + \frac{\varcost^{max}}{{\varcost^{min}}})^{n\frac{\varcost^{min}}{\varcost^{max}}} \leq e^n$, and thus every $n\frac{\varcost^{min}}{\varcost^{max}}$ iteration decreases the potential function by a constant factor ($n$ can be chosen as a smallest positive constant for which $n\frac{\varcost^{min}}{\varcost^{max}} \geq 1$). Furthermore, since the \ac{COS} algorithm starts from an offloading decision vector $\vardecisionsvector^0$ in which all WDs perform computation locally, the potential function begins at the value $\varpotentialfunction(\vardecisionsvector^0) = \sum_{\varplayer \in \varplayersset}\vartime_\varplayer^{ex}$ and cannot drop lower than $\varpotentialfunction^{min}$. Therefore, the \ac{COS} algorithm converges in $\mathcal{O}(n\frac{\varcost^{min}}{\varcost^{max}}\log \frac{\sum_{\varplayer \in \varplayersset}\vartime_\varplayer^{ex}}{\varpotentialfunction^{min}})$ iterations, which proves the result.
\end{proof}
%\color{black}
%From Lemma~\ref{lemm::closed_form_slices}, Lemma~\ref{lemm::closed_form_SRO} and  Theorem~\ref{theo::COS_convergence} we obtain the following result.
%\begin{corollary}\label{corr::SPE_of_ES_COG}
%The \ac{NSO-ERA} has an SPE $(\vardecisionsvector^*,\varpolicy_{\varbandwidthslicecoef}^{*},\varpolicy_{\varuplinkratecoef_{\varAP}}^*,\varpolicy_{\varcomppowercoef_{\varcloud}}^*)$.
%\end{corollary}

In what follows we address the efficiency of the \ac{COS} algorithm in terms of the cost approximation ratio. 

\begin{theorem}\label{theo::approx_offloading_decisions_slices_optimal}
%Consider a vector $\vardecisionsvector^*$ of offloading decisions computed by the \ac{COS} algorithm and an optimal solution $\hat{\vardecisionsvector}$ of~(\ref{obj::offloading})-(\ref{cons::offloading}). Then, 
The \ac{COS} algorithm is a $2.62$-approximation algorithm for the  optimization problem~(\ref{obj::offloading_and_Pb})-(\ref{cons::offloading_and_Pb}) in terms of the system cost, i.e., $\frac{\Tilde{\varcost}(\vardecisionsvector^*)}{\Tilde{\varcost}(\hat{\vardecisionsvector})} \leq 2.62$.
\end{theorem}
\begin{proof}
Let us denote by $\varactivedecisionsset^* \subseteq \varactivedecisionsset$ the set of all vectors of offloading decisions that can be computed using the \ac{COS} algorithm given any policy $\varpolicy_\varbandwidthcoef$ and the collection $(\hat{\varpolicy}_{\varuplinkratecoef_{\varAP}},\hat{\varpolicy}_{\varcomppowercoef_{\varcloud}})$ of the optimal resource allocation policies of slices. Furthermore, let us consider a vector $\vardecisionsvector{}^{*} \in \varactivedecisionsset^*$ and an arbitrary vector $\hat{\vardecisionsvector} \in \varactivedecisionsset$ of offloading decisions. Since there is no WD $\varplayer$ for which the cost $\Tilde{\varcost}_\varplayer(\vardecisionsvector{}^{*})$ can be decreased by unilaterally changing its offloading decision we have the following
 \begin{eqnarray}\label{eq::PoA_per_WD}
  &\hspace{-1cm}\Tilde{\varcost}_\varplayer(\vardecisionsvector{}^{*}) \leq \sum \limits_{\varresource \in \Tilde{\varresourcesset}_{\vardecision{}^{*}_\varplayer} \cap \Tilde{\varresourcesset}_{\hat{\vardecision}_\varplayer}}\!\!\!\!m_\varresource\varcongestionweight_{\varplayer,\varresource}\varcongestionweight_\varresource(\vardecisionsvector{}^{*}) + \!\!\!\!\!\\ \nonumber
  &\hspace{-0.4cm}\sum \limits_{\varresource \in \Tilde{\varresourcesset}_{\vardecision{}^{*}_\varplayer} \setminus \Tilde{\varresourcesset}_{\hat{\vardecision}_\varplayer}} \!\!\!\!\!\!\!\! m_\varresource\big(\varcongestionweight_{\varresource}(\vardecisionsvector{}^{*}) + \varcongestionweight_{\varplayer,\varresource}\big)\varcongestionweight_{\varplayer,\varresource} \leq \!\!\sum \limits_{\varresource \in \Tilde{\varresourcesset}_{\hat{\vardecision}_\varplayer}}\!\!\!\!m_\varresource\big(\varcongestionweight_\varresource(\vardecisionsvector{}^{*}) +\varcongestionweight_{\varplayer,\varresource}\big) \varcongestionweight_{\varplayer,\varresource},\!\!\!\!\!\!\!\!
 \end{eqnarray}
 where $\Tilde{\varresourcesset}_{\vardecision{}^{*}_\varplayer} \subset \Tilde{\varresourcesset}$ and $\Tilde{\varresourcesset}_{\hat{\vardecision}_\varplayer} \subset \Tilde{\varresourcesset}$ denote the the set of resources that WD $\varplayer$ uses in $\vardecisionsvector^*$ and $\hat{\vardecisionsvector}$, respectively.
 By summing (\ref{eq::PoA_per_WD}) over all WDs $\varplayer \in \varactiveplayersset$ and by reordering the summations we obtain 
 % \begin{eqnarray}\label{eq::PoA_all_WD}
  %\Tilde{\varcost}(\vardecisionsvector{}^{*},\varpolicy_\varbandwidthcoef,\varpolicy_\varuplinkratecoef^*,\varpolicy_\varcomppowercoef^*) \!\leq\!\! \sum \limits_{\varplayer \in \varactiveplayersset} \sum \limits_{\varresource \in \Tilde{\varresourcesset}_{\hat{\vardecision}_\varplayer}}\!\!\!\!m_\varresource\big(\varcongestionweight_\varresource(\vardecisionsvector{}^{*}) \!+\!\varcongestionweight_{\varplayer,\varresource}\big) \varcongestionweight_{\varplayer,\varresource}.\!\!
  %\end{eqnarray}
  %By reordering the summations in (\ref{eq::PoA_all_WD}) we obtain 
  \begin{eqnarray}\label{eq::PoA_all_WD_1}
  \Tilde{\varcost}(\vardecisionsvector{}^{*}) \!\leq\!\! \sum \limits_{\varresource \in \varactiveresourcesset} \sum \limits_{\varplayer \in \varoffloaders_\varresource(\hat{\vardecisionsvector})} \!\!\!\!\!\!m_\varresource\big(\varcongestionweight_\varresource(\vardecisionsvector{}^{*})\varcongestionweight_{\varplayer,\varresource} \!+\!\varcongestionweight_{\varplayer,\varresource}^2 \big).\!\!\!\!
  \end{eqnarray}
   From the definition (\ref{eq::weights}) of the total weight $\varcongestionweight_\varresource(\vardecisionsvector)$ on resource $\varresource \!\in\! \Tilde{\varresourcesset}$ and from $\sum \limits_{\varplayer \in \varoffloaders_\varresource(\vardecisionsvector)}\!\varcongestionweight_{\varplayer,\varresource}^2 \leq \varcongestionweight_{\varresource}^{2}(\vardecisionsvector)$ we obtain
   $$\Tilde{\varcost}(\vardecisionsvector{}^{*}) \!\leq\! \sum \limits_{\varresource \in \varactiveresourcesset}\!\!m_\varresource \varcongestionweight_\varresource(\vardecisionsvector{}^{*})\varcongestionweight_\varresource(\hat{\vardecisionsvector}) + \sum \limits_{\varresource \in \varactiveresourcesset}\!\!m_\varresource \varcongestionweight_\varresource^2(\hat{\vardecisionsvector}).
  $$
   Next, let us recall the Cauchy-Schwartz inequality $\sum \limits_{\varresource \in \varactiveresourcesset} a_\varresource b_\varresource \leq \sqrt{\sum \limits_{\varresource \in \varactiveresourcesset} a_\varresource^2 \sum \limits_{\varresource \in \varactiveresourcesset} b_\varresource^2}$. By defining $a_\varresource \triangleq \sqrt{m_\varresource}\varcongestionweight_\varresource(\vardecisionsvector{}^{*})$ and $b_\varresource \triangleq \sqrt{m_\varresource}\varcongestionweight_\varresource(\hat{\vardecisionsvector})$ we obtain the following
  \begin{eqnarray}\label{eq::PoA_all_WD_3}
   &\!\!\!\!\Tilde{\varcost}(\vardecisionsvector{}^{*}) \!\leq\!\! \sqrt{\sum \limits_{\varresource \in \varactiveresourcesset}m_\varresource\varcongestionweight_\varresource^2(\vardecisionsvector{}^{*}) \sum \limits_{\varresource \in \varactiveresourcesset} m_\varresource\varcongestionweight_\varresource^2(\hat{\vardecisionsvector})}\nonumber \\ &\!\!\!\!+ \sum \limits_{\varresource \in \varactiveresourcesset} m_\varresource\varcongestionweight_\varresource^2(\hat{\vardecisionsvector}).
   \end{eqnarray}
   By dividing the right and the left side of (\ref{eq::PoA_all_WD_3}) by $\sum \limits_{\varresource \in \varactiveresourcesset} \varcongestionweight_\varresource^2(\hat{\vardecisionsvector})>0$ and by using (\ref{eq::system_cost_under_optimal_slices}) we obtain
    \begin{eqnarray}\label{eq::PoA_all_WD_4}
   \frac{\Tilde{\varcost}(\vardecisionsvector{}^{*})}{\Tilde{\varcost}(\hat{\vardecisionsvector})} \leq \sqrt{\frac{\Tilde{\varcost}(\vardecisionsvector{}^{*})}{\Tilde{\varcost}(\hat{\vardecisionsvector})}} + 1.
   \end{eqnarray}
   Since (\ref{eq::PoA_all_WD_4}) holds for any vector $\vardecisionsvector{}^{*} \in \varactivedecisionsset^*$ of offloading decisions computed by the \ac{COS} algorithm and for any vector $\hat{\vardecisionsvector} \in \varactivedecisionsset$ of offloading decisions of the WDs, it holds for the worst vector $\vardecisionsvector^* = \argmax \nolimits_{\vardecisionsvector \in \varactivedecisionsset^*}\Tilde{\varcost}(\vardecisionsvector)$ of offloading decisions that can be computed using the \ac{COS} algorithm and for the optimal $\hat{\vardecisionsvector} = \argmin \nolimits_{\vardecisionsvector \in \varactivedecisionsset}\Tilde{\varcost}(\vardecisionsvector)$ solution too. Therefore, by solving (\ref{eq::PoA_all_WD_4}) we obtain that the cost approximation ratio $\frac{\Tilde{\varcost}(\vardecisionsvector^*)}{\Tilde{\varcost}(\hat{\vardecisionsvector})}$ of the \ac{COS} algorithm is upper bounded by $ (3 + \sqrt{5})/2 \cong 2.62$, which proves the theorem.
   \end{proof}

\revo{\begin{theorem}\label{theo::approx_offloading_decisions_slices_optimal_operator_optimal}
%Consider a vector $\vardecisionsvector^*$ of offloading decisions computed by the \ac{COS} algorithm and an optimal solution $\hat{\vardecisionsvector}$ of~(\ref{obj::offloading})-(\ref{cons::offloading}). Then, 
The \ac{COS} algorithm is a $2.62$-approximation algorithm for the  optimization problem~(\ref{obj::offloading})-(\ref{cons::offloading}) in terms of the system cost, i.e., $\frac{\bar{\varcost}(\vardecisionsvector^*)}{\bar{\varcost}(\hat{\vardecisionsvector})} \leq 2.62$.
\end{theorem} 
\begin{proof}
The result can be easily obtained by following the approach used to prove Theorem~\ref{theo::approx_offloading_decisions_slices_optimal}.
\end{proof}}
Finally, from Theorem~\ref{theo::policies_offloading_decomposition} and Theorem~\ref{theo::approx_offloading_decisions_slices_optimal_operator_optimal} we obtain  the approximation ratio bound for the proposed decomposition-based algorithm.
\begin{theorem}\label{theo::PoA_ES_COG}
Given the collection $(\hat{\varpolicy}_{\varbandwidthslicecoef},\hat{\varpolicy}_{\varuplinkratecoef_{\varAP}},\hat{\varpolicy}_{\varuplinkratecoef_{\varcloud}})$ of optimal allocation policies, the proposed decomposition-based algorithm computes a $2.62$-approximation solution to the \ac{JSS-ERM} problem.
\end{theorem}

\section{Numerical Results}
\color{black}
\label{sec::numerical}
\begin{figure*}[tb]
	%\vspace{-0.45cm}
	\begin{minipage}{0.5\textwidth}
		\begin{center}
			\includegraphics[width=\columnwidth]{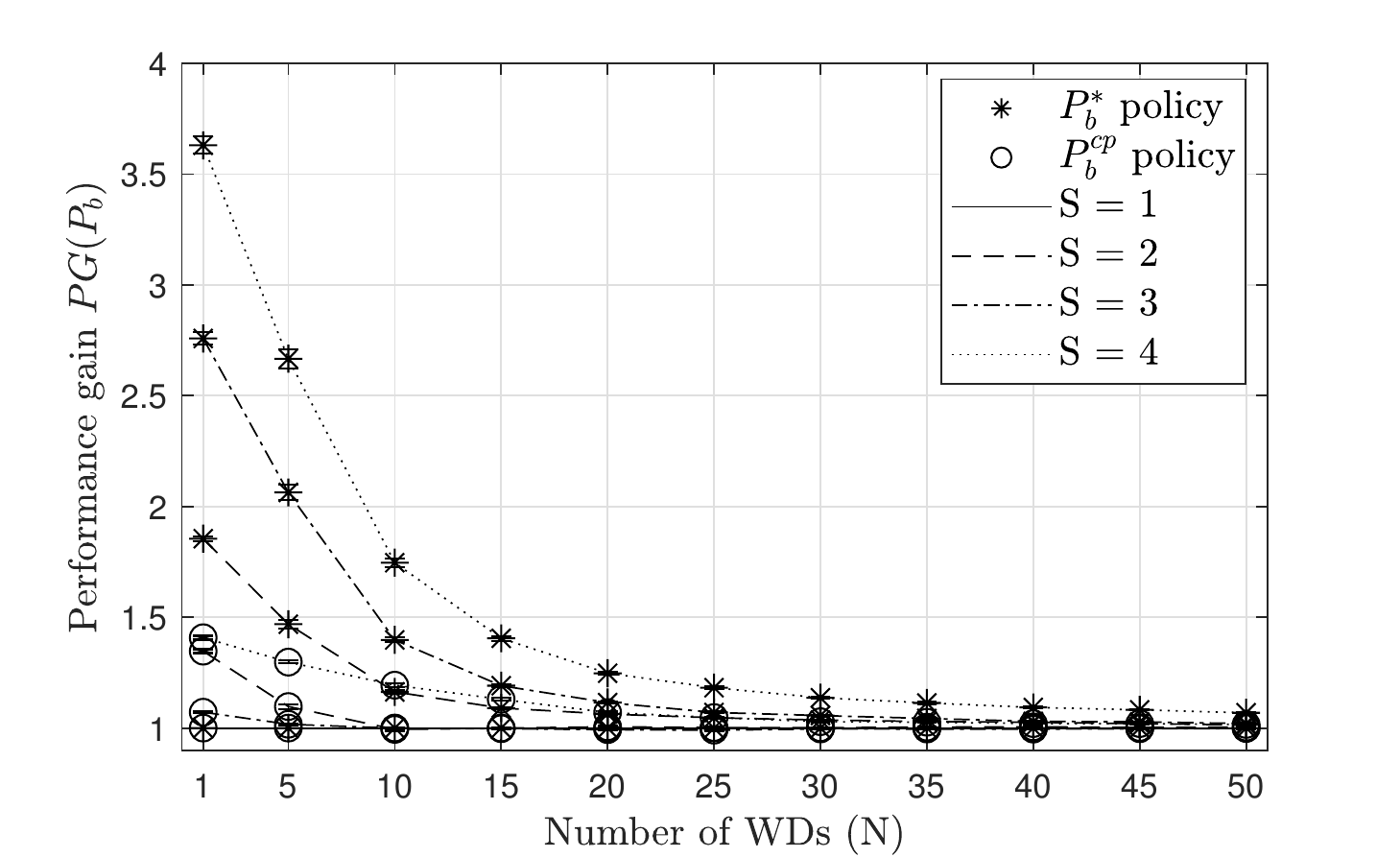}
			%\vspace{-0.8cm}
			\caption{\emph{Performance gain} vs. number of WDs $\varplayerssetdym$.}
			% for $\varAPssetdym = 5$ APs and $\varcloudssetdym = 5$ ECs
			\label{fig::perf_gain}
		\end{center}
	\end{minipage}
	\hspace{0.015\textwidth}
	\begin{minipage}{0.5\textwidth}
		\begin{center}
			\includegraphics[width=\columnwidth]{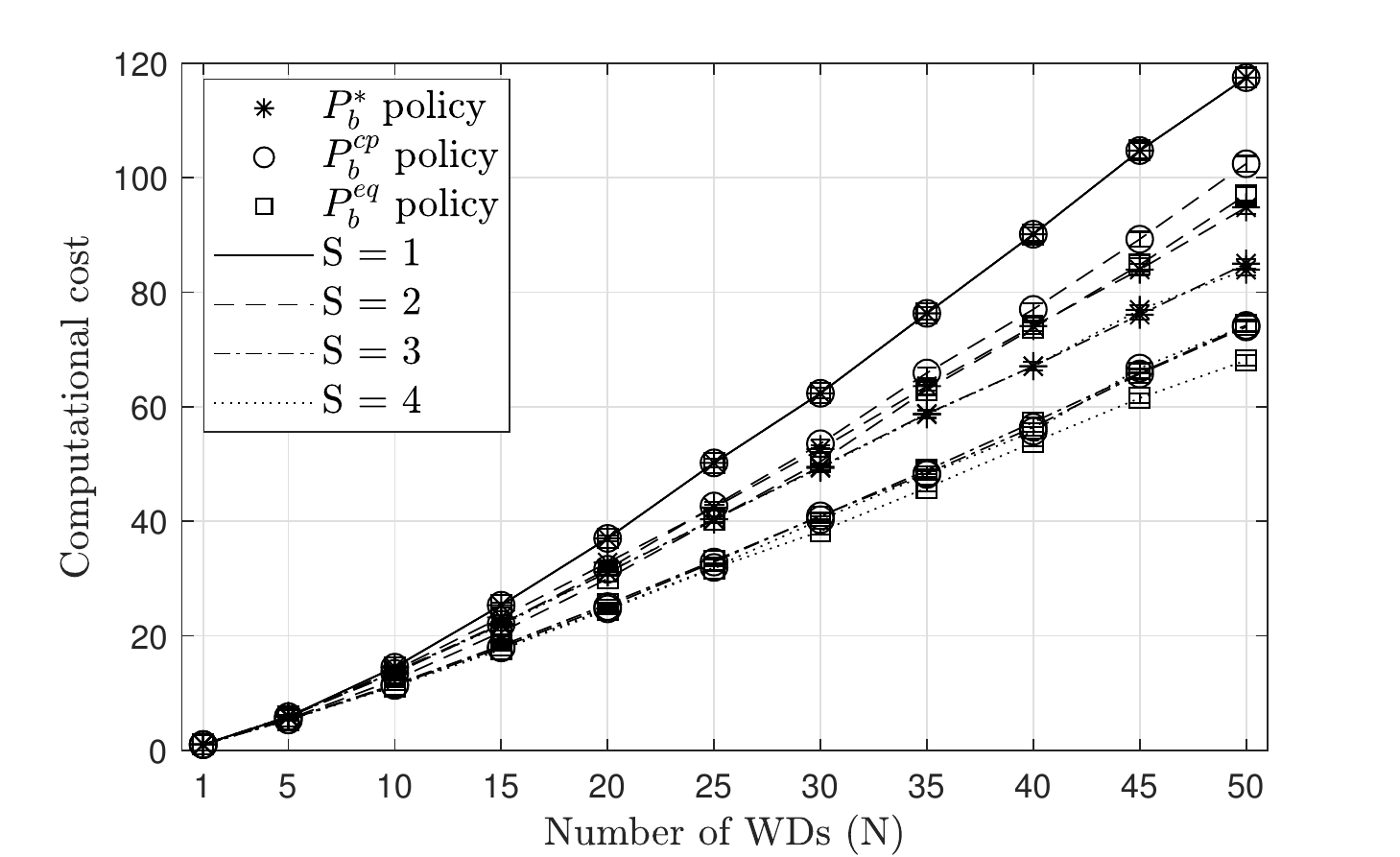}
			%\vspace{-0.8cm}
			\caption{\emph{Computational complexity} vs. number of WDs $\varplayerssetdym$.}
			%for $\varAPssetdym = 5$ APs and $\varcloudssetdym = 5$ ECs
			\label{fig::comp_complexity}
		\end{center}
	\end{minipage}
	\vspace{-0.1cm}
\end{figure*}
We used extensive simulations to evaluate the performance of the proposed resource allocation algorithm. 
%\todo[inline]{Should we shorten the description of the methodology to max 1 column?}
%\todo[inline]{Now it's just a bit longer than 1 column (without baselines :-)). I will shorten the related work and the PoA proof, too. P.S. This year they allow a paper 9 pages long without references.}
To capture the potentially uneven spatial
distribution of ECs, WDs and APs in a dense urban area, we consider a square area of $1 km \times 1 km$ in which WDs and $3$ ECs are placed uniformly at random and $5$ APs are placed at random on a \emph{regular grid} with $25$ points. The channel gain of WD $\varplayer$ to AP $\varAP$ depends on their Euclidean distance $\varplayerdistance_{\varplayer,\varAP}$ and on the path loss exponent $\varpathlossexponent$, which we set to $4$ according to the path loss model in urban and suburban areas~\cite{saunders2007antennas}.
We set the bandwidth $\varbandwidth_{\varAP}$ of $2$ APs to $18$\textit{MHz} and the bandwidth of $3$ APs to $27$\textit{MHz}, corresponding to $25$ and $75$ resource blocks that are $12\times60$\textit{KHz} and $12\times30$\textit{KHz} subcarriers wide~\cite{tsgr2011lte,zaidi20185g}, respectively. We consider that the transmit power $\vartransmissionpower_{\varplayer,\varAP}$ of every WD $\varplayer$ is uniformly distributed on $[10^{-6},0.1]$\textit{W} according to~\cite{lauridsen2014empirical}. We calculate the total thermal noise in
a $\varbandwidth_{\varAP}$\textit{MHz} channel as $\varnoisepower(\text{dBm}) = - 174 + 10 log(\varbandwidth_{\varAP})$ according to~\cite{da2018understanding} and the transmission rate $\varphywirelessrate_{\varplayer,\varAP}$ achievable to WD $\varplayer$ at AP $\varAP$ as $\varphywirelessrate_{\varplayer,\varAP} \!=\! \varbandwidth_{\varAP} log(1 \!+\! \varplayerdistance_{\varplayer,\varAP}^{-\varpathlossexponent}\frac{\vartransmissionpower_{\varplayer,\varAP}}{\varnoisepower})$.

 To set the values for the computational capabilities of the WDs, we consider a line of Samsung Galaxy phones, from the oldest version with $1$ core operating at $1$\textit{GHz} to the one of the newest versions with $8$ cores operating at $2.84$\textit{GHz}. We consider that EC $c_1$ is equipped with $36$ vCPUs operating at $2.3$\textit{GHz} and $96$ vCPUs operating at $3.6$\textit{GHz}. 
 We consider that EC $c_2$ and EC $c_3$ are equipped with $1$ GPU each (with 2048 parallel processing cores operating at $557$\textit{MHz} and 2496 parallel processing cores operating at $560$\textit{MHz}, respectively). 
 %We consider that ECs $c_2$ and EC $c_3$ we are equipped with $1$ NVIDIA Tesla M60 GPU (with 2048 parallel processing cores operating at $557$\textit{MHz}) and $1$ NVIDIA Tesla M60 GPU (with 2496 parallel processing cores operating at $560$\textit{MHz}), respectively. 
 Given the measurements reported in~\cite{codrescu2014hexagon,hackenberg2015energy,Takefuji:2017:GPC:3161341} we assume that a WD, a CPU and a GPU can execute on average $2,3$ and $1$ instructions per cycle (\textit{IPC}), respectively. Based on this, we consider that the computational capability $\varcomppower_\varplayer^\varlocal$ of every WD $\varplayer$ is uniformly distributed on [2, 45.4]\textit{GIPS}, where the lower and the upper bound correspond to the oldest and the newest version of the phone, respectively. Similarly, we calculate the computational capabilities of ECs, and set them to $\varcomppower_{c1} \!=\! 1285.2$\textit{GIPS}, $\varcomppower_{c_2} \!=\! 1140.7$\textit{GIPS} and  $\varcomppower_{c_3} \!=\! 1397.8$\textit{GIPS}.
 %\todo[inline]{Are the WDs more powerful than the ECs?}
 %\todo[inline]{No, this is a typo. I will correct it now!}
 
The input data size $\vardatasize_\varplayer$ is drawn from a uniform distribution on $[1.7, 10]$\textit{Mb} according to measurements in~\cite{fletcher2005road}. The number $\varcomplexityperbit$ of instructions per data bit follows a Gamma distribution~\cite{lorch2004pace} with shape parameter $k = 75$ and scale $\theta = 50$.  Given $\vardatasize_\varplayer$ and $\varcomplexityperbit$, we calculate the complexity of a task as $\vartaskcomplexity_\varplayer \!=\! \vardatasize_\varplayer \varcomplexityperbit$.

Motivated by Amazon EC2 instances~\cite{amazonEC2} designed to support different kinds of applications (e.g., G3 and P2 instances for graphics-intensive and general-purpose GPU applications, and C5 and I3 instances for compute-intensive and non-virtualized workloads), we evaluate the system performance for the following four cases.\\
{\bf \varslicessetdym = 1:}  The slice $\varslice_1$ contains all ECs, and thus is able to support all of the above applications.\\
{\bf \varslicessetdym = 2:} The ECs are sliced such that slice $\varslice_1$ supports the G3.4 instance and slice $\varslice_2$ supports instances C5 and I3.\\
{\bf  \varslicessetdym = 3:} The ECs are sliced such that slices $\varslice_1$ and $\varslice_2$ support P2 and G3s instances, respectively and slice $\varslice_3$ supports instances C5 and I3.\\
{\bf \varslicessetdym = 4:} The ECs are sliced such that slices $\varslice_1, \varslice_2, \varslice_3$ and $\varslice_4$ support P2, G3s, C5 and I3 instances, respectively.

The coefficients \revo{$\frac{1}{\vardeviceslicecoef_{\varplayer,\varslice}}$} were drawn from a continuous uniform distribution on $[0,1]$ and unless otherwise noted, the results are shown for all of the above scenarios.
 
%We use two bandwidth allocation policies of the operator as a basis for comparison. The first policy $\varpolicy_\varbandwidthcoef^{cp}$ shares the bandwidth of each AP $\varAP$ among slices proportionally to the computing resources that slices have and we refer to this policy as \emph{EC proportional slicing}. The second policy $\varpolicy_\varbandwidthcoef^{eq}$ gives the equal share of bandwidth of each AP $\varAP$ to each slice $\varslice$ and we refer to this policy as \emph{equal slicing}. Observe that the \ac{COS} algorithm computes an equilibrium of offloading decisions of WDs for both policies. The results shown are the averages of $300$ simulations, together with $95\%$ confidence intervals.
We use two bandwidth allocation policies $\varpolicy_\varbandwidthcoef$ of the slice orchestrator as a basis for comparison. The first policy $\varpolicy_\varbandwidthcoef^{cp}$ shares the bandwidth of each AP $\varAP$ among slices proportionally to the ECs' resources that slices have. The second policy $\varpolicy_\varbandwidthcoef^{eq}$ gives an equal share of the bandwidth of each AP $\varAP$ to each slice $\varslice$. Observe that the \ac{COS} algorithm computes an \revo{ approximation vector $\vardecisionsvector^*$ of offloading decisions for} both policies (c.f. \revo{Theorem~\ref{theo::COS_convergence} and Theorem~\ref{theo::approx_offloading_decisions_slices_optimal})}. The results
shown are the averages of $300$ simulations, together with $95\%$ confidence intervals.
\begin{figure*}[tb]
	%\vspace{-0.45cm}
    	\begin{minipage}{0.5\textwidth}
		\begin{center}
			\includegraphics[width=\columnwidth]{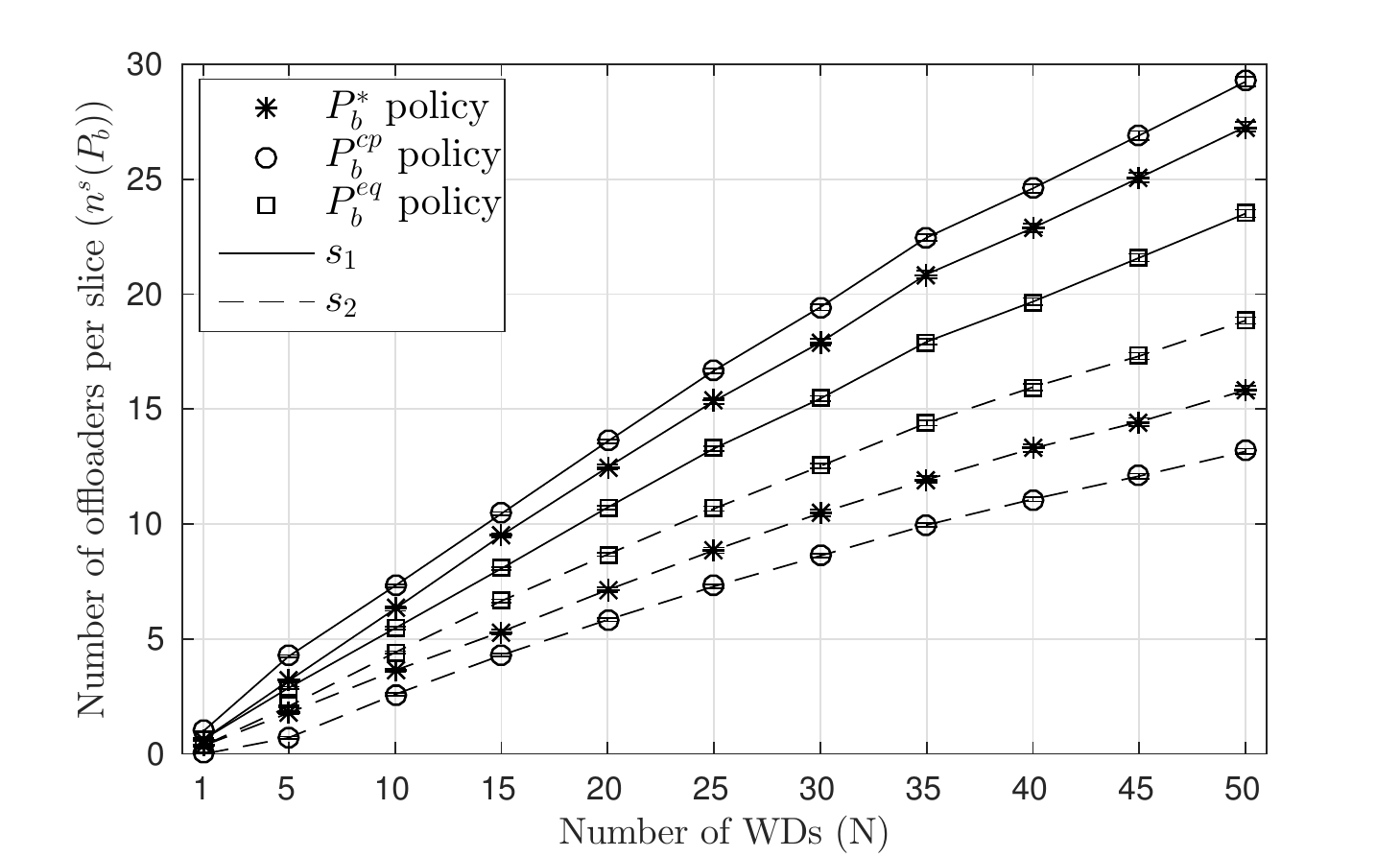}
			%\vspace{-0.8cm}
			\caption{\emph{Number of offloaders per slice} vs. number of WDs $\varplayerssetdym$.}
			%for $\varAPssetdym = 5$ APs and $\varcloudssetdym = 5$ ECs
			\label{fig::number_of_offloaders}
		\end{center}
	\end{minipage}
	\hspace{0.015\textwidth}
		\begin{minipage}{0.5\textwidth}
		\begin{center}
			\includegraphics[width=\columnwidth]{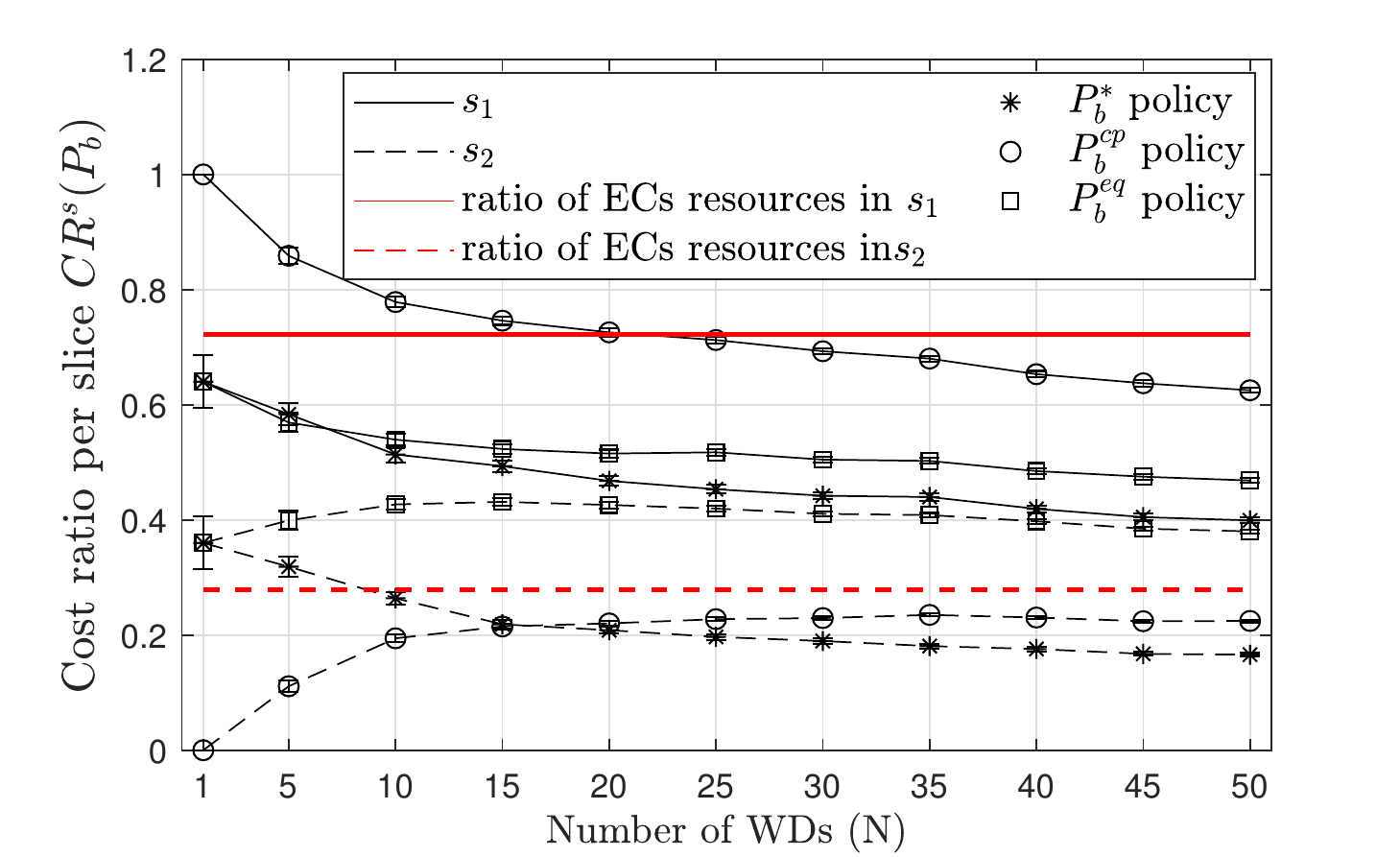}
			%\vspace{-0.8cm}
			\caption{\emph{Cost ratio per slice} vs. number of WDs $\varplayerssetdym$.}
			% for $\varAPssetdym = 5$ APs and $\varcloudssetdym = 5$ ECs
			\label{fig::cost_ratio}
		\end{center}
	\end{minipage}
	\vspace{-0.1cm}
\end{figure*}

\subsection{System Performance}
\label{sec::numerical_sys_perf_gain}
We start with considering the system performance from the point of view of the slice orchestrator. To do so we define the system performance gain $PG(\varpolicy_\varbandwidthcoef)$ for \revo{an inter-slice radio} allocation policy $\varpolicy_\varbandwidthcoef$ w.r.t. the policy $\varpolicy_\varbandwidthcoef^{eq}$ as 
$${PG(\varpolicy_\varbandwidthcoef) = \frac{\varcost(\vardecisionsvector^*,\varpolicy_\varbandwidthcoef^{eq},\hat{\varpolicy}_{\varuplinkratecoef_{\varAP}},\hat{\varpolicy}_{\varcomppowercoef_{\varcloud}})}{\varcost(\vardecisionsvector^*,\varpolicy_\varbandwidthcoef,\hat{\varpolicy}_{\varuplinkratecoef_{\varAP}},\hat{\varpolicy}_{\varcomppowercoef_\varcloud})}.}$$
Figure~\ref{fig::perf_gain} shows $PG(\varpolicy_\varbandwidthcoef)$ as a function of the number $\varplayerssetdym$ of WDs for the optimal $\varpolicy_\varbandwidthcoef^*$ and for the cloud proportional $\varpolicy_\varbandwidthcoef^{cp}$ allocation policy of the operator. We observe that $PG(\varpolicy_\varbandwidthcoef^*) \!=\! PG(\varpolicy_\varbandwidthcoef^{cp}) \!=\! 1$ when $\varslicessetdym \!=\! 1$ because the three solutions are equivalent when there is no slicing. On the contrary, for $\varslicessetdym \!>\! 1$ we observe that $PG(\varpolicy_\varbandwidthcoef^*) \!>\! 1$ and $PG(\varpolicy_\varbandwidthcoef^{cp}) \!>\! 1$, which is due to that the policy $\varpolicy_\varbandwidthcoef^{eq}$ does not take into account that the slices might have different amounts of ECs' resources. We also observe that the policy $\varpolicy_\varbandwidthcoef^*$ achieves better performance gain (up to 2.5 times greater) than the policy $\varpolicy_\varbandwidthcoef^{cp}$ because $\varpolicy_\varbandwidthcoef^*$ \revo{assigns the WDs to slices not only based on the amount of ECs' resources the slices have,}  but also based on how well the slices are tailored for executing their tasks. This effect is especially evident when there are few WDs, because in this case WDs tend to offload their tasks, and thus the system cost is mostly determined by the offloading cost.

\subsection{Computational Cost}
\label{sec::numerical_eq_comp_complexity}
Figure~\ref{fig::comp_complexity} shows the number of \revo{iterations in which the \ac{COS} algorithm computes a decision vector $\vardecisionsvector^*$} as a function of the number $\varplayerssetdym$ of WDs under the optimal $\varpolicy_\varbandwidthcoef^*$, the cloud proportional $\varpolicy_\varbandwidthcoef^{cp}$ and the equal $\varpolicy_\varbandwidthcoef^{eq}$ \revo{inter-slice radio} allocation policy of the slice orchestrator.

 Interestingly, the number of updates decreases with the number $\varslicessetdym$ of slices. \revo{This is due to that the congestion on the logical resources decreases as $\varslicessetdym$ increases, and thus the \ac{COS} algorithm updates the offloading decisions less frequently.} We also observe that the number of updates scales approximately linearly with $\varplayerssetdym$ under all considered policies of the slice orchestrator, and thus we can conclude that the \ac{COS} algorithm is computationally efficient, which makes it a good candidate for computing \revo{an approximation $\vardecisionsvector^*$ to the optimal vector $\hat{\vardecisionsvector}$ of offloading decisions of WDs.}%an equilibrium $\vardecisionsvector^*$ of offloading decisions of WDs in a decentralized way.

\subsection{Performance Within the Slices}
\label{sec::numerical_perf_slices}
We continue with considering the performance from the point of view of the slices. For \revo{an inter-slice radio} allocation policy $\varpolicy_\varbandwidthcoef$, we denote by $\varnumber^\varslice(\varpolicy_\varbandwidthcoef)$ the number of offloaders per slice in \revo{the vector $\vardecisionsvector^*$ of offloading decisions computed by the \ac{COS} algorithm} and we define the cost ratio $CR^{\varslice}(\varpolicy_\varbandwidthcoef)$ per slice w.r.t. the system cost as 
$${CR^{\varslice}(\varpolicy_\varbandwidthcoef) = \frac{\varcost^\varslice(\vardecisionsvector^*,\varpolicy_\varbandwidthcoef,\hat{\varpolicy}_{\varuplinkratecoef_{\varAP}},\hat{\varpolicy}_{\varcomppowercoef_{\varcloud}})}{\varcost(\vardecisionsvector^*,\varpolicy_\varbandwidthcoef,\hat{\varpolicy}_{\varuplinkratecoef_{\varAP}},\hat{\varpolicy}_{\varcomppowercoef_{\varcloud}})}.}$$
Figure~\ref{fig::number_of_offloaders} and Figure~\ref{fig::cost_ratio} show $\varnumber^\varslice(\varpolicy_\varbandwidthcoef)$ and $CR^{\varslice}(\varpolicy_\varbandwidthcoef)$, respectively for the optimal $\varpolicy_\varbandwidthcoef^*$, the cloud proportional $\varpolicy_\varbandwidthcoef^{cp}$ and the equal $\varpolicy_\varbandwidthcoef^{eq}$ \revo{inter-slice radio} allocation policy of the slice orchestrator. The results are shown for $\varslicessetdym = 2$ and the red lines in Figure~\ref{fig::cost_ratio} show the share of the ECs' resources among the slices $\varslice_1$ and $\varslice_2$ (i.e, slices $\varslice_1$ and $\varslice_2$ have approximately $72\%$ and $28\%$ of the resources, respectively). 
%We observe that $\varnumber^{\varslice_1}(\varpolicy_\varbandwidthcoef) > \varnumber^{\varslice_2}(\varpolicy_\varbandwidthcoef)$ and $CR^{\varslice_1}(\varpolicy_\varbandwidthcoef) > CR^{\varslice_2}(\varpolicy_\varbandwidthcoef)$ hold for every policy $\varpolicy_\varbandwidthcoef \in \\varpolicy_\varbandwidthcoef^*, \varpolicy_\varbandwidthcoef^{cp}, \varpolicy_\varbandwidthcoef^{eq}\}$, which is due to that slice $\varslice_1$ has more ECs' resources than slice $\varslice_2$. 
We observe from Figure~\ref{fig::number_of_offloaders} and Figure~\ref{fig::cost_ratio}, respectively that the gap between $\varnumber^{\varslice_1}(\varpolicy_\varbandwidthcoef)$ and $\varnumber^{\varslice_2}(\varpolicy_\varbandwidthcoef)$ and the gap between $CR^{\varslice_1}(\varpolicy_\varbandwidthcoef)$ and $CR^{\varslice_2}(\varpolicy_\varbandwidthcoef)$ are highest in the case of the policy $\varpolicy_\varbandwidthcoef^{cp}$ and lowest in the case of the policy $\varpolicy_\varbandwidthcoef^{eq}$. Therefore, WDs whose tasks are a better match with the EC resources in slice $\varslice_2$ than those in slice $\varslice_1$ cannot fully exploit the ECs' resources in slice $\varslice_2$ under the policy $\varpolicy_\varbandwidthcoef^{cp}$, which allocates bandwidth resources proportionally to the ECs' resources. Similarly, WDs whose tasks are a better match with the EC resources in slice $\varslice_1$ than in slice $\varslice_2$ cannot fully exploit the ECs' resources in slice $\varslice_1$ under the policy $\varpolicy_\varbandwidthcoef^{eq}$, which allocates bandwidth resources equally. On the contrary, the results show that the optimal policy $\varpolicy_\varbandwidthcoef^*$ finds a good match between the EC resources in the slices and the WDs' preferences for different types of computing resources, which makes it a good candidate for dynamic resource management for network slicing coupled with edge computing.
%\todo[inline]{We have to decide what to show here; the ratio or the actual numbers}
\section{Related Work}
\label{sec::related}
Closest to our work a recent game theoretic treatments
of the computation offloading  problem~\cite{ge2012game,wang2013nested,chen2015efficient,jovsilo2017game,jovsilo2019wireless}. In~\cite{ge2012game} the authors considered devices that compete for cloud resources so as to minimize their energy consumption, and proved that an equilibrium of offloading decision can be computed in polynomial time. In~\cite{wang2013nested} the authors considered  devices that maximize their  performance and a profit maximizing service provider, and used backward induction for deriving near optimal strategies for the devices and the operator. In~\cite{chen2015efficient} the authors considered that devices can offload their tasks to a cloud through multiple identical wireless links, modeled the congestion on wireless links, and used a potential function argument for proposing a decentralized algorithm for computing an equilibrium. In~\cite{jovsilo2017game} the authors considered that devices can offload their tasks to a cloud through multiple heterogeneous wireless links, modeled the congestion on wireless and cloud resources, showed that the game played by devices is not a potential game and proposed a decentralized algorithm for computing an equilibrium. In~\cite{jovsilo2019wireless} the authors modeled the interaction between devices and a single network operator as a  Stackelberg game, and provided an algorithm for computing a subgame perfect equilibrium. Unlike these works, we consider the computation offloading problem together with network slicing and we analyze the interaction between \rev{the network operator and the slices.} 

Another line of works considers the network slicing resource allocation problem~\cite{jiang2017network,caballero2017multi,caballero2018network,bega2019deepcog,d2019slice}. %zheng2019elastic
In~\cite{jiang2017network} the authors considered an auction-based model for allocating edge cloud resources to slices and proposed an algorithm for allocating resources to slices so as to maximize the total network revenue. In~\cite{caballero2017multi} the authors considered the radio resources slicing problem and proposed an approximation algorithm for maximizing the sum of the users' utilities. In~\cite{caballero2018network} the authors modeled the interaction between slices that compete for bandwidth resources with the objective to maximize the sum of their users' utilities, and proposed an admission control algorithm under which the slices can reach an equilibrium. 
%In~\cite{zheng2019elastic} the authors considered the problem of slicing of heterogeneous edge resources among multiple service providers and with the objective to achieve a trade-off between the fairness and overall system performance, they proposed a resource allocation criterion for supporting dynamic user demands. 
In~\cite{bega2019deepcog} the authors proposed a deep learning architecture for sharing the resources among network slices in order to meet the users' demand within the slices. In~\cite{d2019slice} the authors considered a radio access network slicing problem and proposed two approximation algorithms for maximizing the total network throughput. Unlike these works, we consider a slicing enabled edge system in \revo{which the slice resource orchestrator assigns WDs to slices and shares radio resources across slices, while the slices manage their own radio and computing resources with the objective to maximize overall system performance.}

To the best of our knowledge ours is the first work to consider slicing and computation offloading to edge clouds jointly, capturing the interaction between \revo{the slice resource orchestrator and the slices.}
\section{Conclusion}
\label{sec::conclusion}
\revo{We have considered the computation offloading problem in an edge computing system under network slicing in which slices jointly manage their own communication and computing resources and the slice resource orchestrator manages communication resources among slices and assigns the WDs to slices. We formulated the problem of minimizing the sum over all WDs' task completion times as a mixed-integer problem, proved that the problem is NP-hard and proposed a decomposition of the problem into a sequence of optimization problems. We proved that the proposed decomposition does not change the optimal solution of the original problem, proposed an efficient approximation algorithm for solving the decomposed problem and proved that the algorithm has bounded approximation ratio. Our numerical results show that the proposed algorithm is computationally efficient. They also show that dynamic allocation of slice resources is essential for maximizing the benefits of edge computing, and slicing could be beneficial for improving overall system performance.}
\begin{comment}
\todo[inline]{Assumptions to be justified:}
\color{red}
\begin{itemize}
    %\item How do we know the average complexity of tasks?
    \item Why do we allow devices to make their own offloading decisions?
    %\item Why do we neglect the time needed to send the results back?
    \item Why do we consider only the optimization of the bandwidth at the network level?
    \item Why do we consider a Stackelberg model in which the operator is leader and the slices are followers?
\end{itemize}
\color{black}
\end{comment}
%\appendix
%\input{appendix_new.tex}

%\section*{Acknowledgment}


%\bibliographystyle{IEEEtran}
\bibliographystyle{IEEEtran}
\bibliography{IEEEabrv,main}
\end{document}